\documentclass[11pt]{article}

\usepackage{amsmath,amssymb,amsthm,mathtools}
\usepackage[margin=1.1in]{geometry}
\usepackage{enumitem}
\usepackage{xcolor}
\usepackage[colorlinks=true,linkcolor=blue!60!black,citecolor=blue!60!black,urlcolor=blue!60!black]{hyperref}

\newcommand{\C}{\mathbb{C}}
\newcommand{\Z}{\mathbb{Z}}
\newcommand{\FF}{\mathbb{F}}
\newcommand{\PP}{\mathbb{P}}
\newcommand{\Pd}{\check{\mathbb{P}}}
\newcommand{\Gr}{\mathrm{Gr}}
\newcommand{\Con}{\mathrm{Con}}
\newcommand{\CC}{\mathrm{CC}}
\newcommand{\Eu}{\mathrm{Eu}}
\newcommand{\Seg}{\mathrm{Seg}}
\DeclareMathOperator{\dc}{dc}
\DeclareMathOperator{\sdc}{sdc}
\DeclareMathOperator{\Sing}{Sing}
\DeclareMathOperator{\codim}{codim}
\DeclareMathOperator{\rk}{rank}
\DeclareMathOperator{\per}{per}
\DeclareMathOperator{\tw}{tw}
\DeclareMathOperator{\pw}{pw}
\DeclareMathOperator{\td}{td}
\DeclareMathOperator{\Circ}{Circ}
\DeclareMathOperator{\dhom}{dhom}
\DeclareMathOperator{\Hom}{Hom}
\newcommand{\dcb}{\overline{\dc}}
\newcommand{\sdcb}{\overline{\sdc}}
\newcommand{\hd}{\check{h}}
\newcommand{\Bbez}{B_{\mathrm{bez}}}
\newcommand{\VP}{\mathsf{VP}}
\newcommand{\VNP}{\mathsf{VNP}}
\newcommand{\VBP}{\mathsf{VBP}}

\theoremstyle{plain}
\newtheorem{theorem}{Theorem}[section]
\newtheorem{proposition}[theorem]{Proposition}
\newtheorem{lemma}[theorem]{Lemma}
\newtheorem{corollary}[theorem]{Corollary}
\newtheorem{conjecture}[theorem]{Conjecture}
\theoremstyle{definition}
\newtheorem{definition}[theorem]{Definition}

\newtheorem{hypothesis}[theorem]{Hypothesis}
\theoremstyle{remark}
\newtheorem{remark}[theorem]{Remark}

\title{Conormal Lower Bounds for Determinantal Complexity:\\ Fermat Slices, Tight Bounds, and a Ceiling}
\author{Karthik Sheshadri\\[3pt]
{\normalsize Independent AI Researcher, San Jose, California, USA}\\
{\normalsize\texttt{karthiksheshadri217@gmail.com}}}
\date{June 2026; revised October 2026}

\begin{document}
\maketitle

\begin{abstract}
The determinantal complexity $\dc(f)$ of a polynomial $f$ is the smallest size of a matrix of affine linear forms whose determinant is $f$, and the border determinantal complexity $\dcb(f)$ is the smallest size of such matrices whose determinants converge to $f$. The polar method bounds these quantities from below by comparing an invariant of the hypersurface $V(f)$, such as the degree of its dual variety, with the number of solutions of a polynomial system built from a determinantal representation of $f$. This paper determines how far the method reaches.

\emph{Lower bounds.} If $f$ has a linear restriction to a form of degree $d$ in $N$ variables with smooth projective zero set, then $\dcb(f)\ge(N-1)(d-1)/(4e)$. Setting the off-diagonal entries of the host matrix to zero turns every homomorphism polynomial $\hom_{F,q}$ with a connected pattern $F$ into a Fermat polynomial, so $\dcb(\hom_{F,q})\ge(q-1)(|E(F)|-1)/(4e)$. For the $\VBP$-complete path family $W_q$ this is tight: $\dcb(W_q)=\Theta(q^3)$. The same argument gives cubic lower bounds for natural $\VP$- and $\VNP$-complete families, and a lower bound in terms of the codimension of the singular locus of $V(f)$.

\emph{A ceiling.} Conversely, the method, together with every refinement in a precisely defined class, cannot prove a lower bound larger than $O(dN)$ for any polynomial of degree $d$ in $N$ variables. The proof computes the solution count exactly --- it is the degree of the dual of a generic determinantal hypersurface, and its $(N-1)$-st root is $(4e+o(1))\,m/N$ --- and shows that every characteristic-cycle invariant of a hypersurface of degree $d$, whatever its singularities, is at most $8(d-1)^{N-1}+8N$. For the $n\times n$ permanent the ceiling is $O(n^3)$. For the power sum $\sum_{i=1}^N x_i^N$ it is $(1+o(1))N^2/(4e)$, which matches the known lower bound.
\end{abstract}

\tableofcontents

\section{Introduction}\label{sec:intro}

\subsection{Background}

The \emph{determinantal complexity} $\dc(f)$ of a polynomial $f\in\C[x_1,\dots,x_N]$ is the least $m$ such that
\[
f(x)\;=\;\det M(x),\qquad M(x)=A_0+\sum_{i=1}^{N}x_iA_i,\quad A_i\in\C^{m\times m}.
\]
The \emph{border determinantal complexity} $\dcb(f)$ is the least $m$ such that $f$ is a limit of such determinants of size $m$, and $\dcb(f)\le\dc(f)$. Valiant's program \cite{Val79} makes the growth of $\dc(\per_n)$ for the $n\times n$ permanent the central question of algebraic complexity: $\dc(\per_n)=n^{\omega(\log n)}$ would separate $\VBP$ from $\VNP$.

Mignon and Ressayre \cite{MR04} proved $\dc(\per_n)\ge n^2/2$ using the rank of the Hessian, Cai, Chen and Li \cite{CCL10} extended this to every characteristic other than $2$, and Landsberg, Manivel and Ressayre \cite{LMR13} proved the same bound for $\dcb(\per_n)$ using dual varieties. Alper, Bogart and Velasco \cite{ABV17} obtained linear bounds from the codimension of the singular locus. On the upper side, Grenet \cite{Gre11} showed $\dc(\per_n)\le 2^n-1$. For about two decades the lower bound stayed quadratic. A preprint of September 2026 \cite{OAI26} gives a cubic lower bound for $\dcb(\per_n)$; its argument uses kernel equations, a Schur-complement reduction and a multihomogeneous solution count of the kind studied in \cite{She26a} and in Section~\ref{sec:count} below.

For the power sum $\sum_{i=1}^N x_i^N$ the best lower bound was linear, $1.5N-3$ \cite{KV21}, while the upper bound is $O(N^2)$ (Remark~\ref{rem:max}). The companion paper \cite{She26a} introduced the \emph{polar method}, described below, and proved\footnote{A quadratic lower bound for the same polynomial also appears in \cite{KV26}; whether its argument falls within Definition~\ref{def:reading} is not analysed here.}
\begin{equation}\label{eq:positive}
\dcb\Bigl(\sum_{i=1}^{N}x_i^{N}\Bigr)\;\ge\;\frac{(N-1)^2}{4e}.
\end{equation}
This paper asks two questions about the method. What does it give for other polynomials, in particular for natural complete families? And how far can it go at all?

\subsection{The polar method in one page}\label{sec:onepage}

Fix a size-$m$ representation $f=\det M$. The method attaches to it a system of polynomial equations, the \emph{kernel system} (Section~\ref{sec:count}). Its unknowns are a point $x$, a kernel vector $v$ and a cokernel vector $u$ of the matrix $M(x)$, and its equations are
\[
M(x)v=0,\qquad u^{\mathsf T}M(x)=0,\qquad\text{plus generic linear conditions on the numbers }u^{\mathsf T}A_iv .
\]
Two facts drive the method.
\begin{enumerate}
\item \emph{The number of solutions is at least an invariant of $f$.} At a point $x$ where $M(x)$ has rank $m-1$, the gradient of $\det M$ is proportional to the vector $(u^{\mathsf T}A_1v,\dots,u^{\mathsf T}A_Nv)$ (Jacobi's formula; Lemma~\ref{lem:condet}). So the extra linear conditions select the points of $V(f)$ whose tangent hyperplane passes through fixed generic points. The number of such points is a classical invariant of $V(f)$, a \emph{polar degree}; for a smooth hypersurface of degree $d$ in $\PP^{N-1}$ it is $d(d-1)^{N-2}$, the degree of the dual hypersurface.
\item \emph{The number of solutions is at most a quantity that depends only on $N$ and $m$.} By B\'ezout's theorem, for \emph{every} $m\times m$ matrix the number of isolated solutions is at most a number $\Bbez(N,m)$, and $\Bbez(N,m)^{1/(N-1)}\le 4e\,m/(N-1)$.
\end{enumerate}
If the invariant exceeds $\Bbez(N,m)$, no representation of size $m$ exists, so $\dc(f)>m$. One can picture a balance: one pan holds an invariant of $f$, the other holds the largest number of solutions that any $m\times m$ matrix can have, and the method proves $\dc(f)>m$ for as long as the first pan is heavier. The comparison is decided after taking $(N-1)$-st roots, and we say \emph{per root} for this scale. For a smooth hypersurface of degree $d$ the first pan weighs about $d-1$ per root and the second about $4e\,m/N$, so the method gives $\dc(f)\gtrsim (N-1)(d-1)/(4e)$. For the power sum, $d=N$ and this is \eqref{eq:positive}. A specialization argument in \cite{She26a} carries the comparison over to limits of determinants, which gives border bounds.

\subsection{Lower bounds from smooth slices}

The argument just sketched uses nothing about the power sum except that its projective zero set is smooth, and it survives linear restrictions: if $g$ is obtained from $f$ by substituting affine linear forms for the variables, then $\dc(g)\le\dc(f)$ and $\dcb(g)\le\dcb(f)$. This gives the following general statement.

\begin{theorem}[Smooth slices; Theorems~\ref{thm:polar-exact} and~\ref{thm:smooth-slice}]\label{thm:intro-slice}
Let $P$ be a polynomial that has a linear restriction equal to a homogeneous form $f(y_1,\dots,y_N)$ of degree $d\ge2$, where $N\ge3$ and $V(f)\subset\PP^{N-1}$ is smooth. Then
\[
\dcb(P)\;\ge\;\frac{N-1}{4e}\bigl(d(d-1)^{N-2}\bigr)^{1/(N-1)}\;\ge\;\frac{(N-1)(d-1)}{4e},
\]
and the symmetric border complexity satisfies the same bound with $2e$ in place of $4e$. For $\dc$ the first inequality holds for every reduced $f$ with $d(d-1)^{N-2}$ replaced by the degree of the dual of $V(f)$.
\end{theorem}

The exact case is proved here in full (Theorem~\ref{thm:polar-exact}); the border and symmetric cases come from \cite{She26a} (Section~\ref{sec:slices} states precisely what is used). The most useful smooth forms are the Fermat polynomials $y_1^d+\dots+y_N^d$, and they appear inside homomorphism polynomials. For a bipartite multigraph $F$ with sides $L$ and $R$ and an integer $q\ge1$, the homomorphism polynomial
\[
\hom_{F,q}(X)=\sum_{\alpha:L\to[q]}\ \sum_{\beta:R\to[q]}\ \prod_{uv\in E(F)} x_{\alpha(u),\beta(v)}
\]
is a polynomial of degree $|E(F)|$ in the $q^2$ entries of a matrix $X$. If $X$ is diagonal, a pair $(\alpha,\beta)$ contributes only if every edge of $F$ receives the same label at both ends; when $F$ is connected this forces all vertices to receive the same label, and
\[
\hom_{F,q}(\mathrm{diag}(y_1,\dots,y_q))=y_1^{|E(F)|}+\dots+y_q^{|E(F)|}.
\]

\begin{theorem}[Homomorphism polynomials; Theorem~\ref{thm:hom-lower}]\label{thm:intro-hom}
Let $F$ be a connected bipartite multigraph with $d=|E(F)|\ge2$ edges and let $q\ge3$. Then
\[
\dcb(\hom_{F,q})\;\ge\;\frac{(q-1)(d-1)}{4e},\qquad \sdcb(\hom_{F,q})\;\ge\;\frac{(q-1)(d-1)}{2e}.
\]
\end{theorem}

Dwivedi, Pago and Seppelt \cite{DPS26} gave criteria in terms of treewidth, pathwidth and treedepth for a family $(\hom_{F_q,q})$ to be $\VBP$-, $\VP$- or $\VNP$-complete. Combined with Theorem~\ref{thm:intro-hom}, they give lower bounds for complete families in each class, and for paths the bound is tight.

\begin{theorem}[A tight bound for a $\VBP$-complete family; Theorem~\ref{thm:vbp}]\label{thm:intro-vbp}
Let $P_L$ be the path with $L$ edges and $W_q=\hom_{P_{q^2},q}$. The family $(W_q)$ is $\VBP$-complete under constant-depth $c$-reductions over fields of characteristic zero, and $\dcb(W_q)=\Theta(q^3)$.
\end{theorem}

We also obtain $\dcb\ge(\frac1{4e}-o(1))q^3$ for a $\VP$-complete family built from complete binary trees and $\dcb\ge(\frac1{2e}-o(1))q^3$ for a $\VNP$-complete family built from $q\times q$ grids (Theorems~\ref{thm:vp} and~\ref{thm:vnp}). Every nonzero polynomial satisfies $\dcb(P)\ge\deg P$, and these families have degree of order $q^2$ in $M=q^2$ variables, so the content of Theorem~\ref{thm:intro-hom} is the extra factor $q=\sqrt M$. The bounds are polynomial and say nothing about separating $\VBP$, $\VP$ and $\VNP$; their interest is that they go beyond the degree bound for natural complete families, and in one case are tight.

Slicing a hypersurface by a generic linear space avoids its singular locus once the slice is small enough, which gives a second general statement.

\begin{theorem}[Codimension principle; Theorem~\ref{thm:codim}]\label{thm:intro-codim}
Let $P\in\C[z_1,\dots,z_M]$ be homogeneous and squarefree of degree $d\ge2$, and let $c$ be the codimension in $\PP^{M-1}$ of the singular locus of $V(P)$, with $c=M$ if $V(P)$ is smooth. If $c\ge3$, then $\dcb(P)\ge(c-1)(d-1)/(4e)$.
\end{theorem}

For the dense family $H_q=\hom_{K_{q,q},q}$, which is $\VNP$-complete, we show that $H_q$ is irreducible for $q\ge3$ and that the singular locus of $V(H_q)$ has codimension between $q$ and $(q-1)^2+1$ (Section~\ref{sec:Hq}). If the upper end is the truth, Theorem~\ref{thm:intro-codim} gives $\dcb(H_q)=\Omega(q^4)$, which by the next part is the most the method can give for $H_q$.

\subsection{The ceiling}

There are two natural ways to strengthen the method: count the solutions more sharply than B\'ezout's theorem does, or replace the polar degree by a finer invariant of $V(f)$, for instance one that also records its singularities. We show that neither takes the method beyond $O(dN)$.

Write $n=N-1$, $d=\deg f$ and $X=V(f)\subset\PP^{n}$. For a generic linear matrix $M(x)=\sum_{i=1}^Nx_iA_i$ of size $m$, let $X_M=V(\det M)$ and let $Z_1\subset\PP^{N-1}\times\PP^{m-1}\times\PP^{m-1}$ be the \emph{kernel incidence} of triples $(x,v,u)$ with $M(x)v=0$ and $u^{\mathsf T}M(x)=0$; $h,a,b$ denote the hyperplane classes of the three factors.

\begin{theorem}[Exact solution count; Section~\ref{sec:count}]\label{thm:intro-count}
For generic $M$, the incidence $Z_1$ maps birationally onto the conormal variety of $X_M$, and
\[
[Z_1]=(h+a)^m\sum_{k=0}^{m-1}(-1)^ka^k(h+b)^{m-1-k}.
\]
The solution count $T(N,m)=\int_{Z_1}(a+b)^{N-2}$ equals the degree of the dual hypersurface of $X_M$, is attained by generic matrices as a count of reduced points, and has explicit closed forms; for instance $T(3,m)=m(m-1)$ and $T(4,m)=m(m-1)^2$. At fixed $N$, $T(N,m)/\Bbez(N,m)\to(N-1)/(2N-3)$ as $m\to\infty$, and for $N\to\infty$, $m/N\to\infty$,
\[
T(N,m)^{1/(N-1)}\;=\;(4e+o(1))\,\frac{m}{N}\;=\;(1+o(1))\,\Bbez(N,m)^{1/(N-1)}.
\]
\end{theorem}

So no sharper counting argument can do better than B\'ezout per root: generic matrices attain $T(N,m)$, and $\Bbez$ exceeds $T$ by a factor of at most about $2$. For the second route we use the \emph{characteristic cycle} of $X$, a standard object that records the polar degrees of $X$ and of each of its singular strata, with multiplicities that measure how singular $X$ is along each stratum (Section~\ref{sec:cc}). For a reduced hypersurface it has the form $(-1)^{\dim X}\sum_S m_S[\Con\bar S]$, with $S$ running over $X$ and its singular strata and $m_S\ge0$, and each conormal variety $\Con\bar S$ has degrees $\delta_0,\dots,\delta_{n-1}$ (its multidegrees). For smooth $X$ the cycle is $[\Con X]$ and $\delta_i=d(d-1)^{n-1-i}$.

\begin{theorem}[Invariant ceiling; Theorem~\ref{thm:ceiling}]\label{thm:intro-ceiling}
For every nonzero homogeneous polynomial $f$ of degree $d\ge2$ in $N$ variables, every stratum $S$ of $X=V(f)_{\mathrm{red}}$ and every $i$,
\[
m_S\,\delta_i(\Con \bar S)\;\le\;8(d-1)^{N-1}+8N .
\]
The same bound, up to factors $e^{O(N)}$, holds for the Chern--Mather, Chern--Schwartz--MacPherson and Milnor classes; for Milnor classes of inputs that are non-reduced or have non-isolated singularities one constant is conditional (Remark~\ref{rem:massey}).
\end{theorem}

A finer invariant carries more information, but it is never a larger number, and only the size of the number matters to the method. Comparing per root, $d-1$ against $4e\,m/N$, the method stops at $m\approx dN/(4e)$.

\begin{theorem}[Ceiling; informal version of Theorem~\ref{thm:nogo}]\label{thm:intro-nogo}
Consider any lower-bound argument for $\dc(f)$ that (a) uses an invariant of $V(f)$ that is a linear combination of characteristic-cycle data (Definition~\ref{def:invariant}), (b) shows that this invariant is at most the number of solutions of a kernel system of the representation, and (c) bounds that number by a quantity valid for every $m\times m$ matrix (Definition~\ref{def:reading}). Then the lower bound it proves is at most $K\,dN$ for an absolute constant $K$. This is unconditional for corank-one kernel systems, except that for Milnor-class invariants of general inputs the constant rests on Remark~\ref{rem:massey}; higher-corank systems are discussed in Remark~\ref{rem:highercorank}.
\end{theorem}

Restricting first does not help: a linear restriction of $f$ in $N'\le N$ variables has degree at most $d$, so applying the method to it gives at most $K\,dN'$. In particular, Theorem~\ref{thm:intro-slice} extracts from a smooth slice in $N$ variables everything that the method can extract from it, up to a constant factor.

\subsection{Consequences}

\paragraph{The permanent.} For $\per_n$, with $d=n$ and $N=n^2$, the ceiling is $O(n^3)$ (Corollary~\ref{cor:perm}): no argument covered by Theorem~\ref{thm:intro-nogo} proves that $\dc(\per_n)$ is superpolynomial. The cubic lower bound of \cite{OAI26} shows that this order is attained for the permanent, so what remains open for such arguments is the constant. Generic slices do not suffice for this: the singular locus of the permanent hypersurface has codimension at most $2n$, so the codimension principle gives at most $(\frac1{2e}+o(1))n^2$ (Proposition~\ref{prop:perm-codim}).

\paragraph{The power sum.} Here $d=N$, and the quadratic order of the ceiling is already forced by the upper bound $\dc(\sum_ix_i^N)=O(N^2)$; no sound method can prove more. The content is the constant: for invariants of the type used so far, the method cannot prove more than $(1+o(1))N^2/(4e)$ (Corollary~\ref{cor:diagonal}), which matches \eqref{eq:positive}. If $\dc(\sum_ix_i^N)$ is larger than $N^2/(4e)$, this method cannot show it.

\paragraph{Homomorphism polynomials.} For $W_q$ the method gives the true value $\Theta(q^3)$ from a slice in $q$ variables. For the dense family $H_q$ (degree $q^2$ in $q^2$ variables) the ceiling is $O(q^4)$, and the codimension principle would reach it if the singular locus of $V(H_q)$ is as small as Conjecture~\ref{conj:topdim} predicts.

\subsection{What the ceiling theorem does not cover}\label{sec:intro-open}

Theorem~\ref{thm:intro-nogo} is a statement about arguments of the kind (a)--(c), made precise in Definitions~\ref{def:invariant} and~\ref{def:reading}. It says nothing about: the Hessian-rank argument of \cite{MR04}; the dimension arguments of \cite{LMR13,ABV17}, which use coarse functions of the characteristic cycle but are not of the form (a)--(c) (Section~\ref{sec:barriers}); rank methods such as shifted partial derivatives, which are capped by different arguments \cite{EGOW18,ELSW18}; scheme-theoretic invariants of the dual variety, such as its Hilbert function or Castelnuovo--Mumford regularity; and any use of a determinantal representation that does not go through counting solutions of a kernel system. In particular, it does not say that $\dc(\sum_ix_i^N)$ or $\dc(\per_n)$ is small.

\subsection{Organization}

Section~\ref{sec:prelim} fixes notation and lists the facts about characteristic cycles that are used; readers without background in singularity theory can take them as black boxes. Section~\ref{sec:count} studies the kernel system and computes its solution count; it uses only intersection theory on products of projective spaces. Section~\ref{sec:slices} proves the smooth-slice bounds and the codimension principle, and Section~\ref{sec:hom} applies them to homomorphism polynomials. Section~\ref{sec:ceiling} bounds the invariants, and Section~\ref{sec:nogo} defines the class of arguments covered and proves the ceiling theorem. Section~\ref{sec:barriers} compares with other methods, and Section~\ref{sec:open} lists open problems. Appendix~\ref{app:examples} checks the formulas on small examples, and Appendix~\ref{app:comp} describes the computer calculations of Section~\ref{sec:hom}.

\paragraph{A note on verification.} The closed forms of Section~\ref{sec:count} were checked by exact integer arithmetic against brute-force expansion for $3\le N\le 12$ and $2\le m\le 6N$, and against independent geometric computations (Appendix~\ref{app:examples}). The transform of Section~\ref{sec:ceiling} was derived by exact linear algebra for $n\le 8$ and is proved in closed form for all $n$. Statements whose proof rests on a citation rather than a self-contained argument say so.

\section{Preliminaries}\label{sec:prelim}

\subsection{Notation}

We work over $\C$. Throughout, $f$ is the target polynomial, $N$ is its number of variables, $d=\deg f$, and $m$ is the size of a determinantal representation. In Sections~\ref{sec:prelim}, \ref{sec:count}, \ref{sec:ceiling} and~\ref{sec:nogo} we write $n=N-1$ for the dimension of the projective space $\PP^n=\PP(\C^N)$ in which $X=V(f)$ lives. Elsewhere $n$ is the size of a permanent or a parameter of a family, and this is always said explicitly. For a homogeneous polynomial $P$ we write $V(P)$ for its projective zero set and $\Sing V(P)$ for the set of points where $P$ and all its first partial derivatives vanish; $V(P)$ is \emph{smooth} if $\Sing V(P)=\emptyset$. A polynomial is \emph{squarefree} if it has no repeated irreducible factor.

\subsection{Determinantal complexity}

The \emph{symmetric} variants $\sdc(P)$ and $\sdcb(P)$ are defined like $\dc$ and $\dcb$, with all $A_i$ required to be symmetric. The closure in the definition of $\dcb$ may be taken in the Zariski or the Euclidean topology; the two closures of the set of determinants of size $m$ coincide, since this set is constructible. We say that $Q$ is a \emph{linear restriction} of $P$ if $Q$ is obtained from $P$ by substituting affine linear forms in new variables for the variables of $P$.

\begin{lemma}\label{lem:monotone}
If $Q$ is a linear restriction of $P$, then $\dc(Q)\le\dc(P)$, $\dcb(Q)\le\dcb(P)$, and the same holds for $\sdc$ and $\sdcb$. For every nonzero $P$, $\dcb(P)\ge\deg P$.
\end{lemma}

\begin{proof}
Substituting affine linear forms into a determinant of affine linear forms gives another such determinant of the same size, symmetric if the original was, and the substitution is continuous. A determinant of size $m$ of affine linear forms has degree at most $m$, and polynomials of degree at most $m$ form a closed subspace.
\end{proof}

\begin{remark}[Homogenization]\label{rem:homog}
For homogeneous $f$ of degree $d$, a \emph{homogeneous} representation $f=\det(\sum_ix_iA_i)$ forces $m=d$, so all of the size slack lives in the affine-linear model. It is recovered geometrically as follows. If $f=\det(A_0+\sum_ix_iA_i)$ with $m\ge d$, then
\[
G(x_0,x)\;:=\;\det\Bigl(x_0A_0+\sum_i x_iA_i\Bigr)\;=\;\sum_{j}x_0^{\,m-j}f_j(x)\;=\;x_0^{\,m-d}f(x),
\]
since the degree-$j$ part $f_j$ of $f$ vanishes for $j\ne d$. So a size-$m$ affine-linear representation of $f$ is the same thing as a size-$m$ homogeneous linear representation of $x_0^{m-d}f$ in $N+1$ variables, whose reduced zero locus $V(x_0f)_{\mathrm{red}}\subset\PP^N$ has degree at most $d+1$. The kernel-system geometry of Section~\ref{sec:count} is therefore run at arbitrary size $m$ on the form $x_0^{m-d}f$, never on $f$ itself, and Lemma~\ref{lem:homog-profile} controls the multiplicity $x_0^{m-d}$.
\end{remark}

\subsection{Conormal varieties and multidegrees}\label{sec:conormal}

Let $\Pd^n$ be the dual projective space of hyperplanes, and let $h\in A^1(\PP^n)$ and $\hd\in A^1(\Pd^n)$ be the hyperplane classes; on a product we keep the same letters for their pullbacks. For an irreducible closed subvariety $Z\subseteq\PP^n$, the \emph{conormal variety} is
\[
\Con Z\;=\;\overline{\{(z,H)\in \PP^n\times\Pd^n\;:\;z\in Z_{\mathrm{sm}},\ T_zZ\subseteq H\}},
\]
the closure of the set of pairs (smooth point of $Z$, hyperplane tangent to $Z$ there). It is irreducible of dimension $n-1$ for every $Z$. The \emph{dual variety} $Z^\vee\subseteq\Pd^n$ is its image under the second projection, and in characteristic zero the involution $(z,H)\mapsto(H,z)$ carries $\Con Z$ to $\Con Z^\vee$ \cite{GKZ94,Tev05}. For an $(n-1)$-dimensional cycle $C$ on $\PP^n\times\Pd^n$, its \emph{multidegrees} are
\[
\delta_i(C)\;=\;\int_C h^i\,\hd^{\,n-1-i},\qquad i=0,\dots,n-1,
\]
and we call $i$ the \emph{slot} of $\delta_i$ and $n-i$ its \emph{codegree}. When $C=\Con X$ for a hypersurface $X$, these are the classical \emph{polar degrees} of $X$ \cite{Pie78}: $\delta_{n-1}(\Con X)=\deg X$, and $\delta_0(\Con X)=\deg X^\vee$ when $X^\vee$ is a hypersurface (and $0$ otherwise). For $X$ smooth of degree $d$, $\delta_i(\Con X)=d(d-1)^{n-1-i}$. We write $\delta_{\mathrm{top}}(X):=\delta_0(\Con X)$; it is the number of points of $X$ whose tangent hyperplane contains $n-1$ generic points, the classical \emph{class} of $X$. Multidegrees of effective cycles are nonnegative, since $h$ and $\hd$ are globally generated.

\subsection{Constructible functions and characteristic cycles}\label{sec:cc}

This subsection can be read as a black box. We use exactly three facts about characteristic cycles. \emph{(i) Linearity and the smooth case:} $\varphi\mapsto\CC(\varphi)$ is linear, and for a smooth closed $W$ it sends the indicator function of $W$ to $\pm[\Con W]$. \emph{(ii) Slices determine degrees:} the Euler characteristics of the generic linear sections of $X$ determine the multidegrees of $\CC(\mathbf 1_X)$; Theorem~\ref{thm:transform} gives the explicit formula. \emph{(iii) Positivity:} for a reduced hypersurface $X$, all coefficients of $(-1)^{\dim X}\CC(\mathbf 1_X)$ are nonnegative, so no single coefficient can exceed the total. The rest of this subsection states these precisely.

A constructible function $\varphi$ on $\PP^n$ is a finite sum $\sum_W c_W\mathbf{1}_W$, $c_W\in\Z$, over closed subvarieties. We write $\int_A\varphi\,d\chi$ for the Euler-characteristic integral over a constructible set $A$, and for $0\le k\le n$ we define the \emph{generic slice values}
\[
\chi_k(\varphi)\;:=\;\int_{L_k}\varphi\,d\chi,\qquad L_k\subset\PP^n\ \text{a generic linear subspace of dimension }k,
\]
which are well defined by Kleiman transversality; we set $\chi_{-1}=\chi_{-2}=0$. For $\varphi=\mathbf{1}_X$, $\chi_k(\varphi)=\chi(X\cap L_k)$.

MacPherson's local Euler obstruction $\Eu_W$ \cite{Mac74} is a constructible function supported on $W$ with $\Eu_W\equiv 1$ on $W_{\mathrm{sm}}$; the functions $\Eu_W$ form a basis of the group of constructible functions. The \emph{characteristic cycle} is the linear map $\varphi\mapsto\CC(\varphi)$ from constructible functions to Lagrangian cycles in $T^*\PP^n$ normalized by
\[
\CC(\Eu_W)\;=\;(-1)^{\dim W}\,[\,\overline{T^*_W\PP^n}\,] ,
\]
where $\overline{T^*_W\PP^n}$ is the closure of the conormal bundle of $W_{\mathrm{sm}}$; after projectivization it corresponds to $[\Con W]$ for $W\subsetneq\PP^n$ and to the zero section for $W=\PP^n$. With this normalization, $\CC$ of a constructible complex is $\CC$ of its stalkwise Euler characteristic and Kashiwara--Dubson index theory applies \cite{Kas73,BDK81,KS90}. We write
\[
\CC(\varphi)\;=\;n_0(\varphi)\,[\text{zero section}]\;+\;\sum_{W\subsetneq\PP^n} n_W(\varphi)\,[\Con W],
\qquad
\Delta_i(\varphi)\;:=\;\sum_{W} n_W(\varphi)\,\delta_i(\Con W).
\]
Two standard facts are used. First, the \emph{index theorem} \cite{Kas73,BDK81,Dub84,KS90}: in the projective-slice form we need, each $\chi_k(\varphi)$ is a universal $\Z$-linear function of the vector $(n_0,\Delta_0,\dots,\Delta_{n-1})$ of $\varphi$, with coefficients that depend only on $n$ and $k$; this is the topological Radon transform of Ernstr\"om and Matsui--Takeuchi \cite{Ern94,MT07}. Second, \emph{Kashiwara positivity}: the characteristic cycle of a perverse sheaf is effective \cite{Kas73,KS90,Gin86}. A reduced hypersurface $X\subset\PP^n$ is a local complete intersection, so $\mathbb{Q}_X[n-1]$ is perverse \cite{DimST}; its stalkwise Euler characteristic is $(-1)^{n-1}\mathbf{1}_X$, and therefore
\begin{equation}\label{eq:positivity}
(-1)^{n-1}\,\CC(\mathbf{1}_X)\;=\;\sum_{S} m_S\,[\Con \bar S],\qquad m_S\ge 0,
\end{equation}
a nonnegative combination over the closures of the strata of a Whitney stratification of $X$, with $m_X=1$ on the open stratum and no zero-section term.

\subsection{Nearby cycles and Milnor data}\label{sec:nearby}

For a nonzero homogeneous $f$ of degree $d$ with $X=V(f)_{\mathrm{red}}$, let
\[
\nu_f(x)\;:=\;\chi\bigl(F_{f,x}\bigr),\qquad F_{f,x}\ \text{the Milnor fibre at $x$ of a local equation of $f$},
\]
which is well defined since local equations differ by units. It is sensitive to multiplicity: for $f=g^e$ with $g$ reduced and smooth at $x\in V(g)$, $\nu_f(x)=e$. It is the stalkwise Euler characteristic of the nearby-cycle complex, and since shifted nearby cycles of a perverse sheaf are perverse \cite{KS90,DimST}, $(-1)^{n-1}\CC(\nu_f)$ is effective. We write $\varphi_f:=\nu_f-\mathbf{1}_X$ for the vanishing-cycle variant. Milnor classes and the Chern--Mather and Chern--Schwartz--MacPherson classes are $\Z$-combinations of the data $n_W\delta_i$ attached to $\CC(\mathbf{1}_X)$, $\CC(\Eu_X)$ and $\CC(\nu_f)$, with coefficient masses at most $2^{O(N)}$ \cite{Mac74,Alu18}; Section~\ref{sec:nogo} makes this the definition of the class of invariants covered by the ceiling theorem.

\section{The kernel system and its solution count}\label{sec:count}

This section computes, for a generic $m\times m$ linear matrix, the number of solutions of the kernel system exactly. Since generic matrices attain it, it is a lower bound for any bound on the count that is valid for all matrices.

\subsection{The kernel incidence}

Fix $M(x)=\sum_{i=1}^Nx_iA_i$ with $A_i\in\C^{m\times m}$, and write $X_M=V(\det M)\subset\PP^{N-1}$. For $1\le c\le m-1$ the \emph{corank-$c$ kernel incidence} is
\[
Z_c\;=\;\bigl\{(x,V,U)\in\PP^{N-1}\times\Gr(c,\C^m)\times\Gr(c,\C^m)\;:\;M(x)V=0,\ U^{\mathsf T}M(x)=0\bigr\}.
\]
We mostly use $c=1$ and write points of $Z_1$ as $(x,v,u)\in P:=\PP^{N-1}\times\PP^{m-1}_v\times\PP^{m-1}_u$, with $h,a,b$ the three hyperplane classes. ``Generic $M$'' means that the linear map $x\mapsto M(x)$ from $\C^N$ to $\C^{m\times m}$ lies in a dense open set.

\begin{lemma}[Genericity]\label{lem:generic}
Let $N\ge 3$, $m\ge 2$, and let $M$ be generic. Then:
\begin{enumerate}
\item $X_M$ is reduced and irreducible, and $\Sing X_M=\Sigma_2:=\{x: \rk M(x)\le m-2\}$, which is empty or of pure codimension $4$ in $\PP^{N-1}$.
\item $Z_1$ has pure dimension $N-2$; over $X_M\setminus\Sigma_2$ the projection $Z_1\to X_M$ is an isomorphism, and the preimage $W$ of $\Sigma_2$ has dimension at most $N-3$.
\item $Z_1$ is reduced and irreducible, and smooth away from $W$.
\item For every $x\in\PP^{N-1}$, the linear map $\C^N\to\Hom(\ker M(x),\,\C^m/\mathrm{im}\,M(x))$ sending $e_i$ to the map induced by $A_i$ is surjective. In particular, for every $(x,v,u)\in Z_1$ some $u^{\mathsf T}A_iv$ is nonzero.
\end{enumerate}
\end{lemma}

\begin{proof}[Proof sketch]
In the space $\PP^{m^2-1}$ of all matrices, the determinant hypersurface $D$ is irreducible of degree $m$, and $\Sing D$ is the locus of rank at most $m-2$, of codimension $4$. The locus $S_c$ of matrices of corank exactly $c$ is smooth of codimension $c^2$, and its normal space at a matrix $B$ is $\Hom(\ker B,\C^m/\mathrm{im}\,B)$, the identification sending a tangent vector $E$ to $v\mapsto Ev \bmod \mathrm{im}\,B$ \cite{GKZ94}. The universal corank-one incidence in $\PP^{m^2-1}\times\PP^{m-1}\times\PP^{m-1}$ is smooth and irreducible: it fibres over $\PP^{m-1}\times\PP^{m-1}$ with fibres the linear spaces $\{B: Bv=0,\ u^{\mathsf T}B=0\}$, cut out by $2m-1$ independent linear conditions.

If $N\le m^2$, a generic $M$ is a generic linear subspace of matrices, and Kleiman's transversality theorem makes it transverse to $D$, to $\Sing D$, to every $S_c$ and to the universal incidence; this gives (1)--(3), and transversality to $S_c$ at $M(x)$ is exactly the surjectivity in (4). If $N>m^2$, the map $x\mapsto M(x)$ is surjective, $X_M$ is the cone over $D$ with vertex $\PP(\ker)$, all pullbacks are transverse off the vertex, and at the vertex $M(x)=0$ and (4) holds because the map is onto; the preimage of the vertex in $Z_1$ has dimension $N-(m-1)^2-2\le N-3$. The last sentence of (4) holds because $u\ne0$ pairs nontrivially with $\C^m/\mathrm{im}\,M(x)$, so the functional $\phi\mapsto u^{\mathsf T}\phi(v)$ on $\Hom(\ker M(x),\C^m/\mathrm{im}\,M(x))$ is nonzero. Over the corank-one locus, kernel and cokernel are lines, which gives the isomorphism in (2).
\end{proof}

The next lemma is the Schur-complement step of \cite[$\S$3.1]{She26a}. It holds for every $M$, not only generic ones.

\begin{lemma}[Local structure at corank one]\label{lem:schur}
Let $M(x)=\sum_ix_iA_i$ be any linear matrix and let $p\in\PP^{N-1}$ with $\rk M(p)=m-1$, with kernel line $[v_0]$ and cokernel line $[u_0]$. Then the projection $Z_1\to V(\det M)$ restricts to an isomorphism of schemes from a neighbourhood of $(p,v_0,u_0)$ in $Z_1$ onto a neighbourhood of $p$ in $V(\det M)$.
\end{lemma}

\begin{proof}
Choose bases with $v_0=u_0=e_1$. Near $p$ write $M=\begin{pmatrix}\alpha&\beta\\ \gamma&\Delta\end{pmatrix}$ with $\alpha$ a scalar and $\Delta$ an invertible $(m-1)\times(m-1)$ block. In the charts $v=(1,w)$ and $u=(1,w')$, the equations $Mv=0$ and $u^{\mathsf T}M=0$ are equivalent to
\[
w=-\Delta^{-1}\gamma,\qquad w'^{\mathsf T}=-\beta\Delta^{-1},\qquad s:=\alpha-\beta\Delta^{-1}\gamma=0 .
\]
Since $\det M=\det\Delta\cdot s$, the ideal of $Z_1$ in this chart is generated by the graph equations for $w,w'$ and by $\det M/\det\Delta$, where $\det\Delta$ is a unit.
\end{proof}

\begin{lemma}[$Z_1$ is the conormal variety of the determinant]\label{lem:condet}
Let $M$ be generic, $N\ge3$, $m\ge2$. The map
\[
\psi:Z_1\to \PP^{N-1}\times\Pd^{N-1},\qquad
(x,v,u)\;\longmapsto\;\bigl(x,\ [\,u^{\mathsf T}A_1v:\cdots:u^{\mathsf T}A_Nv\,]\bigr)
\]
is a morphism, it maps $Z_1$ birationally onto $\Con(X_M)$, and $\psi^*\hd=a+b$. Hence
\[
\delta_i(\Con X_M)\;=\;\int_{Z_1}h^i\,(a+b)^{\,N-2-i},\qquad i=0,\dots,N-2 .
\]
In the universal case $N=m^2$ this recovers the classical fact that the dual of the determinant hypersurface is the Segre variety of rank-one matrices \cite{GKZ94,Tev05}.
\end{lemma}

\begin{proof}
By Lemma~\ref{lem:generic}(4) the coordinates of $\psi$ never vanish simultaneously, so $\psi$ is a morphism. Jacobi's formula gives $d(\det M)=\mathrm{tr}\bigl(\mathrm{adj}(M)\,dM\bigr)$. At a corank-one matrix $M(x)$ with kernel $\C v$ and cokernel $\C u^{\mathsf T}$ the adjugate has rank one, $\mathrm{adj}(M(x))=\kappa\, vu^{\mathsf T}$ with $\kappa\ne 0$, so
\[
\partial_i\det M(x)=\kappa\,u^{\mathsf T}A_iv\qquad(1\le i\le N).
\]
By Lemma~\ref{lem:generic}(1) this gradient is nonzero on $X_M\setminus\Sigma_2=(X_M)_{\mathrm{sm}}$ and spans the conormal line there. So $\psi$ maps the dense open set $Z_1\setminus W$ bijectively onto the part of $\Con X_M$ lying over $(X_M)_{\mathrm{sm}}$; since $Z_1$ is irreducible and projective, $\psi(Z_1)=\Con X_M$ and $\psi$ is birational. A linear form $\ell$ on $\Pd^{N-1}$ pulls back to $u^{\mathsf T}(\sum_i\ell_iA_i)v$, a section of $\mathcal O(0,1,1)$; hence $\psi^*\hd=a+b$. The multidegree formula is the projection formula, since $\psi_*[Z_1]=[\Con X_M]$.
\end{proof}

\begin{remark}[Why the ceiling theorem has the scope it has]\label{rem:reading-sees}
Lemma~\ref{lem:condet} is the structural reason for the scope of Section~\ref{sec:nogo}: counting solutions of the corank-one system against multihomogeneous conditions in $(h,a,b)$ is intersection theory on a birational model of the conormal variety of $X_M$. Whatever invariant of $f$ a method extracts this way, it sees it through the cycle-theoretic shadow of $\Con$. Corank-$c$ incidences for $c\ge 2$ see conormal data of the deeper degeneracy loci; Proposition~\ref{prop:monotone} addresses them under a stated hypothesis, and none of the corank-one results depend on it.
\end{remark}

\subsection{The class of the incidence}

\begin{theorem}[Class of $Z_1$]\label{thm:class}
For generic $M$, in $A^\bullet(P)$,
\[
[Z_1]\;=\;(h+a)^m\cdot c_{m-1}(K),\qquad
c_{m-1}(K)\;=\;\sum_{k=0}^{m-1}(-1)^k\,a^k\,(h+b)^{\,m-1-k},
\]
where $K=\ker\bigl(\mathcal O(h+b)^{\oplus m}\xrightarrow{\;\cdot v\;}\mathcal O(h+a+b)\bigr)$ is the rank-$(m-1)$ kernel bundle of the contraction with the tautological vector $v$.
\end{theorem}

\begin{proof}
The $m$ entries of $M(x)v$ are bilinear in $(x,v)$, so they form a section of $E_1:=\mathcal O(h+a)^{\oplus m}$ on $\PP^{N-1}\times\PP^{m-1}_v$. For generic $M$ its zero locus $Y:=\{Mv=0\}$ fibres over $X_M$ with the projectivized kernels as fibres, generically points, so $\dim Y=N-2$. This is the expected codimension $m$ in the $(N+m-3)$-dimensional product, so the section is regular and $[Y]=c_m(E_1)=(h+a)^m$.

On $Y\times\PP^{m-1}_u$, the $m$ entries of $u^{\mathsf T}M(x)$ form a section $s$ of $\mathcal O(h+b)^{\oplus m}$. The coordinates of $v$ define a surjection $\mathcal O(h+b)^{\oplus m}\to\mathcal O(h+a+b)$, $(\phi_j)\mapsto\sum_j\phi_jv_j$, with kernel $K$ of rank $m-1$. On $Y\times\PP^{m-1}_u$ the section $s$ takes values in $K$, because $\sum_j(u^{\mathsf T}M)_jv_j=u^{\mathsf T}(Mv)=0$. Its zero locus is $Z_1$, of the expected codimension $m-1$ by Lemma~\ref{lem:generic}(2), and for generic $M$ the section is regular. Hence $[Z_1]=[Y\times\PP^{m-1}_u]\cdot c_{m-1}(K)=(h+a)^m\,c_{m-1}(K)$.

From $0\to K\to\mathcal O(h+b)^{\oplus m}\to\mathcal O(h+a+b)\to0$ we get $c(K)=(1+h+b)^m/(1+h+a+b)$. Write $s=h+b$ and $t=h+a+b=s+a$. In $\Z[s,a]$,
\begin{equation}\label{eq:telescope}
\Bigl(\sum_{k=0}^{m-1}(-1)^ka^k s^{\,m-1-k}\Bigr)(s+a)\;=\;s^m-(-a)^m ,
\end{equation}
while the degree-$(m-1)$ part of $(1+s)^m\sum_{j\ge0}(-1)^jt^j$ is $\sum_{j=0}^{m-1}(-1)^j\binom{m}{m-1-j}s^{m-1-j}t^{j}$, and
\begin{align*}
\Bigl(\sum_{j=0}^{m-1}(-1)^j\binom{m}{m-1-j}s^{m-1-j}t^{j}\Bigr)\,t
&=s^m-\sum_{j=0}^{m}\binom{m}{j}s^{m-j}(-t)^{j}\\
&=s^m-(s-t)^m\;=\;s^m-(-a)^m .
\end{align*}
Both expressions for $c_{m-1}(K)$ give the same product with $t$, which is not a zero divisor in $\Z[h,a,b]$, so they agree.
\end{proof}

\subsection{The solution count}

By Lemma~\ref{lem:condet}, the multidegrees of the conormal variety of the generic determinantal hypersurface are read off $Z_1$. Define
\[
T_i(N,m)\;:=\;\int_{P}\,[Z_1]\,h^i(a+b)^{N-2-i}\;=\;\delta_i(\Con X_M),\qquad T(N,m):=T_0(N,m)=\deg X_M^{\vee}.
\]

\begin{theorem}[Closed forms]\label{thm:T}
With $[\,\cdot\,]$ denoting coefficient extraction in $\Z[h,a,b]$:
\begin{align}
T(N,m)&=\bigl[h^{N-1}a^{m-1}b^{m-1}\bigr]\,(h+a)^m\Bigl(\sum_{k=0}^{m-1}(-1)^ka^k(h+b)^{m-1-k}\Bigr)(a+b)^{N-2}
\label{eq:T-poly}\\
&=\bigl[h^{N-1}a^{m-1}b^{m-1}\bigr]\,\frac{(h+a)^m(h+b)^m(a+b)^{N-2}}{h+a+b}
\label{eq:T-rational}\\
&=\sum_{j=0}^{N-2}\binom{N-2}{j}\sum_{k\ge0}(-1)^k\binom{m}{k+j+1}\binom{m-1-k}{\,N-2-j-k\,}
\label{eq:T-binom}\\
&=\sum_{j=0}^{N-2}\binom{N-2}{j}\sum_{l=0}^{j}(-1)^{j-l}\binom{m}{l}\binom{m+j-l}{\,N-1-l\,}.
\label{eq:T-binom2}
\end{align}
$T(N,m)$ is a polynomial in $m$ of degree $N-1$ with leading coefficient $\binom{2N-4}{N-2}/(N-1)!$. Moreover $T_i(N,m)=T(N-i,m)$ for $0\le i\le N-2$, with the convention $T(2,m):=m$.
\end{theorem}

\begin{proof}
\eqref{eq:T-poly} is Theorem~\ref{thm:class}, since an intersection number on $P$ is the coefficient of $h^{N-1}a^{m-1}b^{m-1}$. In \eqref{eq:T-rational} the fraction is expanded at $h=\infty$, $1/(h+a+b)=\sum_{s\ge0}(-1)^s(a+b)^sh^{-s-1}$. By \eqref{eq:telescope}, the integrand of \eqref{eq:T-poly} differs from the fraction in \eqref{eq:T-rational} by $(-a)^m(h+a)^m(a+b)^{N-2}/(h+a+b)$, all of whose terms have $a$-degree at least $m$; so the two coefficients agree.

For \eqref{eq:T-binom}, expand $(h+a)^m=\sum_\alpha\binom m\alpha h^{m-\alpha}a^\alpha$, $(a+b)^{N-2}=\sum_j\binom{N-2}ja^jb^{N-2-j}$ and $c_{m-1}(K)=\sum_k(-1)^ka^k\sum_\beta\binom{m-1-k}\beta h^{m-1-k-\beta}b^\beta$. Matching exponents forces $\alpha=m-1-k-j$ and $\beta=m-N+1+j$, after which the $h$-degree is automatically $N-1$, and the coefficient is $\binom{m}{k+j+1}\binom{m-1-k}{N-2-j-k}$.

For \eqref{eq:T-binom2}, fix $j$, put $r=N-2-j$, and write $\binom{m-1-k}{r-k}=[z^{r-k}](1+z)^{m-1-k}$. Then the inner sum of \eqref{eq:T-binom} is $[z^r](1+z)^{m-1}\sum_k\binom m{k+j+1}w^k$ with $w=-z/(1+z)$, and $\sum_k\binom m{k+j+1}w^k=w^{-(j+1)}\bigl((1+w)^m-\sum_{l\le j}\binom ml w^l\bigr)$ with $1+w=(1+z)^{-1}$. Substituting,
\[
\text{inner}_j=(-1)^{j+1}[z^{N-1}](1+z)^j+\sum_{l=0}^{j}(-1)^{j-l}\binom ml\,[z^{N-1-l}](1+z)^{m+j-l},
\]
and the first term vanishes since $j<N-1$. Each $\binom{m}{l}\binom{m+j-l}{N-1-l}$ is a polynomial in $m$ of degree $N-1$ with leading coefficient $1/(l!\,(N-1-l)!)$, so
\[
(N-1)!\cdot[\text{lead}]=\sum_{j}\binom{N-2}{j}\sum_{l=0}^j(-1)^{j-l}\binom{N-1}{l}
=\sum_j\binom{N-2}j^2=\binom{2N-4}{N-2},
\]
using the alternating partial row sum $\sum_{l\le j}(-1)^{j-l}\binom{N-1}{l}=\binom{N-2}{j}$ and Vandermonde. Finally, the same argument as for \eqref{eq:T-rational} gives $T_i(N,m)=[h^{N-1}a^{m-1}b^{m-1}]\,h^i(h+a)^m(h+b)^m(a+b)^{N-2-i}/(h+a+b)$, and dividing out $h^i$ gives the expression for $T(N-i,m)$; for $i=N-2$ the value is $\deg X_M=m$.
\end{proof}

\subsection{Classical anchors}

\begin{proposition}[Anchors]\label{prop:anchors}
$T(3,m)=m(m-1)$, $T(4,m)=m(m-1)^2$ and $T(5,5)=220$. The first two are the dual degrees of the smooth generic determinantal plane curve and surface, and the third satisfies Teissier's identity
\[
T(5,5)\;=\;\underbrace{5\cdot4^3}_{\text{smooth value}}\;-\;\underbrace{50}_{\#\text{nodes}}\cdot\underbrace{(\mu+\mu^{(1)})}_{=2}\;=\;320-100 .
\]
\end{proposition}

\begin{proof}
The values follow from \eqref{eq:T-binom2}: at $N=3$ both inner sums equal $\binom m2$, giving $2\binom m2=m(m-1)$; at $N=4$ the inner sums are $\binom m3$, $\tfrac{m(m-1)(2m-1)}6$, $\binom m3$ with weights $1,2,1$, summing to $m(m-1)^2$; at $(5,5)$ they are $5,35,35,5$ with weights $1,3,3,1$, giving $220$. By Lemma~\ref{lem:generic}(1), $\Sigma_2$ is empty for $N\le4$, so $X_M$ is a smooth hypersurface of degree $m$ in $\PP^{N-1}$ and $\delta_0=m(m-1)^{N-2}$. For $(N,m)=(5,5)$, $\Sigma_2\subset\PP^4$ is the rank-$\le3$ locus of a generic map $\mathcal O(-1)^{\oplus5}\to\mathcal O^{\oplus5}$, of expected codimension $(5-3)^2=4$, so it is a finite set of ordinary nodes of the quintic threefold $X_M$. By Thom--Porteous their number is the coefficient of $h^4$ in $c_2^2-c_1c_3$ with $c=(1-h)^{-5}$, i.e.\ $15^2-5\cdot35=50$. Each ordinary node of a threefold lowers the dual degree by $\mu+\mu^{(1)}=1+1=2$ \cite{Tei73,Pie78}, so $\deg X_M^\vee=320-100=220$.
\end{proof}

\subsection{The count is attained}

\begin{proposition}[The count is attained]\label{prop:witness}
Let $M$ be generic and $0\le i\le N-2$. For $i$ generic linear forms on $\PP^{N-1}$ and $N-2-i$ generic linear forms $\ell^{(1)},\dots,\ell^{(N-2-i)}$ on $\Pd^{N-1}$, the system
\begin{gather*}
M(x)v=0,\qquad u^{\mathsf T}M(x)=0,\qquad \text{the $i$ linear conditions on $x$},\\
u^{\mathsf T}\Bigl(\textstyle\sum_k\ell^{(t)}_kA_k\Bigr)v=0\qquad (1\le t\le N-2-i)
\end{gather*}
has exactly $T_i(N,m)$ solutions in $P$, all reduced and all with $\rk M(x)=m-1$. Consequently every bound on the number of isolated solutions of this system that holds for all $M$ of size $m$ is at least $T_i(N,m)$.
\end{proposition}

\begin{proof}
The conditions are generic members of two linear systems on $Z_1$: the hyperplane system $|h|$ and the system spanned by the sections $u^{\mathsf T}A_kv$ of $\mathcal O(a+b)$. Both are base-point-free on $Z_1$, the second by Lemma~\ref{lem:generic}(4). By Bertini's theorem in characteristic zero, generic members of base-point-free systems cut every subvariety in the expected dimension and meet the smooth locus transversally. So the $N-2$ conditions miss $W$, which has dimension at most $N-3$, and meet $Z_1\setminus W$, which is smooth, transversally in finitely many reduced points. Their number is $\int[Z_1]h^i(a+b)^{N-2-i}=T_i(N,m)$. A bound that holds for all $M$ holds for the generic one.
\end{proof}

Appendix~\ref{app:conic} checks the case $(N,m)=(3,2)$ by hand: two reduced solutions, matching $T(3,2)=2=\deg(\text{conic})^\vee$, together with the single extra B\'ezout point predicted by Proposition~\ref{prop:bez} below.

\subsection{Comparison with the B\'ezout estimate}

The companion paper \cite{She26a} bounds the count by the multihomogeneous B\'ezout number of the system with one dependent equation discarded:
\begin{align*}
\Bbez(N,m)&:=\bigl[h^{N-1}a^{m-1}b^{m-1}\bigr]\,(h+a)^m(h+b)^{m-1}(a+b)^{N-2}\\
&\;=\sum_{j=0}^{N-2}\binom{N-2}{j}\binom{m}{j+1}\binom{m-1}{\,N-2-j\,}.
\end{align*}

\begin{proposition}[The B\'ezout excess is one-sided]\label{prop:bez}
For generic $M$ and a generic choice of discarded equation, the system with one equation of $u^{\mathsf T}M=0$ discarded cuts out $Z_1\cup Z'$ with $Z'$ effective of pure dimension $N-2$, and
\[
\Bbez(N,m)\;=\;T(N,m)\;+\;\int [Z']\,(a+b)^{N-2}\;\ge\;T(N,m).
\]
At $N=3$, $\Bbez(3,m)=\binom m2+m(m-1)=\tfrac32\,T(3,m)$ for every $m$.
\end{proposition}

\begin{proof}
Replace $u^{\mathsf T}M=0$ by the $m-1$ equations $u^{\mathsf T}Mw_k=0$, where $w_1,\dots,w_{m-1}$ span a generic hyperplane $H_w\subset\C^m$. Where $v\notin H_w$, these equations and $Mv=0$ imply $u^{\mathsf T}M=0$, because $u^{\mathsf T}Mv=0$ and $H_w+\C v=\C^m$. So the new zero scheme is $Z_1$ together with a residual part $Z'$ lying over $\{v\in H_w\}$. A dimension count as in Lemma~\ref{lem:generic} shows that the new scheme has pure expected dimension $N-2$, so its class is $(h+a)^m(h+b)^{m-1}$, and $Z'$ is effective. Intersecting with $(a+b)^{N-2}$, which is globally generated, gives the inequality. The $N=3$ evaluation is the two-term formula.
\end{proof}

\subsection{Asymptotics}

The asymptotics rest on a closed form with nonnegative terms, which removes every alternating-sum issue from the estimates.

\begin{lemma}[A positive form of the count]\label{lem:T-pos}
For all $N\ge3$ and $m\ge2$,
\begin{equation}\label{eq:T-pos}
T(N,m)\;=\;\sum_{i=\lceil N/2\rceil}^{\min(m,\,N-1)}\binom{m}{i}\binom{i-1}{N-1-i}\binom{2i-2}{i-1}.
\end{equation}
In particular $T(N,m)$ is nondecreasing in $m$, and $T(N,m)\ge\binom{m}{N-1}\binom{2N-4}{N-2}$ when $m\ge N-1$.
\end{lemma}

\begin{proof}
Since $(h+a)(h+b)=h(h+a+b)+ab$,
\[
\frac{(h+a)^m(h+b)^m}{h+a+b}\;=\;\sum_{i=1}^m\binom{m}{i}\,h^i\,(h+a+b)^{i-1}(ab)^{m-i}\;+\;\frac{(ab)^m}{h+a+b},
\]
and the last term contains only negative powers of $h$, so it does not contribute to \eqref{eq:T-rational}. In the $i$-th term, multiplied by $(a+b)^{N-2}$, the factor $(ab)^{m-i}$ forces the extraction of $h^{N-1-i}a^{i-1}b^{i-1}$ from $(h+a+b)^{i-1}(a+b)^{N-2}$. Choosing $h^{N-1-i}$ from $(h+a+b)^{i-1}$ can be done in $\binom{i-1}{N-1-i}$ ways and leaves $(a+b)^{2i-2}$, whose $a^{i-1}b^{i-1}$-coefficient is $\binom{2i-2}{i-1}$. Terms with $i<\lceil N/2\rceil$ or $i>N-1$ vanish. (Checks: at $N=3$ the single term $i=2$ gives $m(m-1)$; at $N=4$ the terms give $2\binom m2+6\binom m3=m(m-1)^2$; at $(5,5)$ they give $120+100=220$. The identity was also checked against \eqref{eq:T-binom2} by exact arithmetic for $3\le N\le12$, $2\le m\le6N$.)
\end{proof}

\begin{theorem}[Asymptotics]\label{thm:asymp}
Fix $N\ge3$. As polynomials in $m$,
\[
T(N,m)=\frac{\binom{2N-4}{N-2}}{(N-1)!}\,m^{N-1}\bigl(1+O_N(1/m)\bigr),
\qquad
\Bbez(N,m)=\frac{\binom{2N-3}{N-1}}{(N-1)!}\,m^{N-1}\bigl(1+O_N(1/m)\bigr),
\]
so that $T(N,m)/\Bbez(N,m)\to(N-1)/(2N-3)$ as $m\to\infty$. For every $N\ge3$ and $m\ge2$,
\begin{equation}\label{eq:bez-upper}
T(N,m)^{1/(N-1)}\;\le\;\Bbez(N,m)^{1/(N-1)}\;\le\;\frac{4e\,m}{N-1},
\end{equation}
and in any regime with $N\to\infty$ and $m/N\to\infty$,
\[
T(N,m)^{1/(N-1)}\;=\;(4e+o(1))\,\frac mN\;=\;\bigl(1+o(1)\bigr)\,\Bbez(N,m)^{1/(N-1)} ,
\]
where the $o(1)$ terms depend only on $N$ and $m/N$.
\end{theorem}

\begin{proof}
The leading coefficient of $T$ is Theorem~\ref{thm:T}. For $\Bbez$, every term $\binom{N-2}j\binom{m}{j+1}\binom{m-1}{N-2-j}$ is nonnegative and $\binom mr\le m^r/r!$, so
\[
\Bbez(N,m)\;\le\;m^{N-1}\sum_{j=0}^{N-2}\frac{\binom{N-2}j}{(j+1)!\,(N-2-j)!}\;=\;\frac{\binom{2N-3}{N-1}}{(N-1)!}\,m^{N-1},
\]
by Vandermonde, and the same computation gives the leading coefficient. The ratio of the two leading coefficients is $(N-1)/(2N-3)$. Using $T\le\Bbez$ (Proposition~\ref{prop:bez}), $\binom{2N-3}{N-1}\le2^{2N-3}$, $(N-1)!\ge((N-1)/e)^{N-1}$ and $2^{(2N-3)/(N-1)}\le4$ gives \eqref{eq:bez-upper}. For the lower bound, keep the term $i=N-1$ of Lemma~\ref{lem:T-pos}: for $m\ge N-1$, using $\binom{2K}K\ge4^K/(2K+1)$ with $K=N-2$ and $(N-1)!\le e\sqrt{N-1}\,((N-1)/e)^{N-1}$,
\begin{align*}
T(N,m)^{1/(N-1)}&\ge\Bigl(\frac{(m-N+2)^{N-1}}{(N-1)!}\cdot\frac{4^{N-2}}{2N-3}\Bigr)^{1/(N-1)}\\
&\ge\Bigl(1-\frac Nm\Bigr)\,\frac{4e\,m}{N-1}\,\Bigl(4e(2N-3)\sqrt{N-1}\Bigr)^{-1/(N-1)},
\end{align*}
where the first error factor depends only on $m/N$ and the second only on $N$. Combining with \eqref{eq:bez-upper} and $\Bbez\ge T$ gives the last display.
\end{proof}

So B\'ezout's estimate is loose by a factor of about $2$ and tight per root. For bounded ratios $m/N$ we only need a crude but explicit lower bound.

\begin{proposition}[Lower bound at bounded ratio]\label{prop:floor}
For all $N\ge3$ and $m\ge2N$,
\[
T(N,m)\;\ge\;\Bigl(\frac{2m}{N}\Bigr)^{N-1}.
\]
\end{proposition}

\begin{proof}
For $N=3$ and $N=4$ the closed forms give the claim for all $m\ge2$: $m(m-1)\ge(2m/3)^2$ is equivalent to $5m\ge9$, and $m(m-1)^2\ge(m/2)^3$ to $2\sqrt2\,(m-1)\ge m$. For $N\ge5$, Lemma~\ref{lem:T-pos} with $\binom mk\ge(m/k)^k$ and $\binom{2K}K\ge4^K/(2K+1)$ gives $T(N,m)\ge(m/(N-1))^{N-1}4^{N-2}/(2N-3)$, so it suffices that $2^{N-3}\ge(2N-3)(1-1/N)^{N-1}$. Since $(1-1/N)^{N-1}\le e^{-4/5}<\tfrac9{20}$ for $N\ge5$, it is enough that $2^{N-3}\ge\tfrac9{20}(2N-3)$, which holds at $N=5$ and is preserved when $N$ increases by one.
\end{proof}

\begin{remark}[Numerical values]\label{rem:floor-numerics}
On the grid $3\le N\le12$, $2N\le m\le6N$, the minimum of $T(N,m)^{1/(N-1)}N/m$ is $2.7386\ldots$, attained at $(N,m)=(3,6)$ with $T=30$, and at each fixed $N$ this quantity increases with $m$. A saddle-point analysis of \eqref{eq:T-rational} suggests $T^{1/(N-1)}=\theta(\rho)\,(m/N)(1+o(1))$ at $m=\rho N$, with $\theta$ increasing and $\theta(\rho)\to4e$; we do not prove this, and nothing below uses it.
\end{remark}

\subsection{Higher corank}

The corank-one results above are unconditional. For corank $c\ge2$ we record the natural monotonicity statement, conditional on one explicit hypothesis whose verification is Schubert-calculus bookkeeping that we have not carried out.

\begin{hypothesis}[Effective specialization]\label{hyp:S}
For each $c\ge2$ and each size $m$, there is a flat degeneration of the slot-$0$ corank-$c$ system (Definition~\ref{def:reading}) --- specializing the $c$-planes to contain a fixed corank-one pair $(v,u)$ --- under which the limit of the solution class is the class of the slot-$0$ corank-one system of size $m-c+1$ \emph{plus an effective cycle}.
\end{hypothesis}

\begin{proposition}[Higher corank is not cheaper; conditional on Hypothesis~\ref{hyp:S}]\label{prop:monotone}
Assume Hypothesis~\ref{hyp:S} and let $c\ge2$. Every bound on the number of isolated solutions of the slot-$0$ corank-$c$ system that holds for all $M$ of size $m$ is at least $T(N,m-c+1)$. In particular, if $c=o(m)$, then in the regime $N\to\infty$, $m/N\to\infty$ every such bound $B$ satisfies $B^{1/(N-1)}\ge(4e-o(1))\,m/N$.
\end{proposition}

\begin{proof}[Discussion]
By Fulton--Pragacz \cite{FP98,Pra88}, the class of the corank-$c$ incidence is a nonnegative combination of Schubert monomials in the tautological classes, and intersection numbers of effective cycles with globally generated classes are nonnegative. Granting Hypothesis~\ref{hyp:S}, the corank-$c$ count dominates the corank-one count at size $m-c+1$, which is attained by generic matrices (Proposition~\ref{prop:witness}). The asymptotic statement is Theorem~\ref{thm:asymp}. What is missing is the sign of the correction terms produced by the specialization: Fulton--Pragacz positivity gives effectivity of the corank-$c$ classes themselves, but not of the specialized decomposition.
\end{proof}

\section{Lower bounds from smooth slices}\label{sec:slices}

This section proves the lower-bound side of the method for general targets. The exact case is proved in full from the results of Section~\ref{sec:count}; the border case is taken from \cite{She26a}.

\subsection{The exact model}

\begin{theorem}[Polar-degree bound]\label{thm:polar-exact}
Let $f$ be a reduced homogeneous form of degree $d\ge2$ in $N\ge3$ variables. If $f=\det(A_0+\sum_{i=1}^Ny_iA_i)$ with $A_i\in\C^{m\times m}$, then
\[
\delta_{\mathrm{top}}(V(f))\;\le\;\Bbez(N,m).
\]
Consequently, every polynomial $P$ that has a linear restriction equal to $f$ satisfies
\[
\dc(P)\;\ge\;\frac{N-1}{4e}\,\delta_{\mathrm{top}}(V(f))^{1/(N-1)},
\]
and if $V(f)$ is smooth, $\dc(P)\ge\frac{N-1}{4e}\bigl(d(d-1)^{N-2}\bigr)^{1/(N-1)}\ge\frac{(N-1)(d-1)}{4e}$.
\end{theorem}

\begin{proof}
This is the argument of \cite{She26a}, written for a general target. Put $M(x_0,y)=x_0A_0+\sum_iy_iA_i$, a linear matrix in $N+1$ variables with $\det M=G:=x_0^{m-d}f$ (Remark~\ref{rem:homog}), and let $X=V(f)\subset\PP^{N-1}$. Let $Z_1\subset P:=\PP^N\times\PP^{m-1}\times\PP^{m-1}$ be the kernel incidence of $M$, and consider the system
\[
Mv=0,\qquad u^{\mathsf T}M=0,\qquad \lambda(x_0,y)=0,\qquad u^{\mathsf T}\Bigl(\textstyle\sum_{k=0}^N\ell^{(t)}_kA_k\Bigr)v=0\quad(1\le t\le N-2),
\]
with generic $\lambda,\ell^{(1)},\dots,\ell^{(N-2)}\in\C^{N+1}$. Write $\ell'^{(t)}=(\ell^{(t)}_1,\dots,\ell^{(t)}_N)$.

\emph{Points from the target.} Let $\Pi\subset X$ be the set of points whose tangent hyperplane contains the $N-2$ generic points $[\ell'^{(t)}]$. By the definition of $\delta_{\mathrm{top}}$ and Bertini's theorem on $\Con X$, $\Pi$ consists of $\delta_{\mathrm{top}}(X)$ points, all smooth points of $X$. For $[y]\in\Pi$, the line of $\PP^N$ joining $[1:0:\dots:0]$ to $[0:y]$ lies in $V(G)$ and meets $\{\lambda=0\}$ in a single point $p=[x_0:ty]$ with $x_0\ne0$ and $t\ne0$, for generic $\lambda$.

\emph{Each point lifts to a solution.} At $p$ the coordinate $x_0$ is a unit, so near $p$ the hypersurface $V(G)$ is the cone over $X$, which is smooth there, and $\partial G/\partial y_i(p)=x_0^{m-d}\partial_if(ty)$ is not identically zero. By Jacobi's formula as in Lemma~\ref{lem:condet}, $M(p)$ has corank exactly one, and with kernel $v$ and cokernel $u$ the vector $(u^{\mathsf T}A_kv)_{k=0}^N$ is a nonzero multiple of $\nabla G(p)$, whose $x_0$-component $(m-d)x_0^{m-d-1}f(ty)$ vanishes. So the contraction conditions at $\hat p=(p,v,u)$ say exactly that $[\ell'^{(t)}]\in T_{[y]}X$, which holds. Thus each of the $\delta_{\mathrm{top}}(X)$ points $p$ gives a solution $\hat p$.

\emph{Each lift is isolated.} By Lemma~\ref{lem:schur}, near $\hat p$ the incidence $Z_1$ is isomorphic to $V(G)$ near $p$, that is, to the cone over $X$. Under this isomorphism the contraction conditions cut out the cone over the points of $\Pi$ near $[y]$, a single line, and $\lambda=0$ cuts that line in the single point $p$. So $\hat p$ is an isolated solution.

\emph{B\'ezout.} Replace $u^{\mathsf T}M=0$ by the $m-1$ equations $u^{\mathsf T}Mw_k=0$, where the $w_k$ span a generic hyperplane $H_w\subset\C^m$ containing none of the finitely many kernel vectors $v$ above. As in the proof of Proposition~\ref{prop:bez}, near each $\hat p$ the new system has the same solutions, so $\hat p$ stays isolated. The new system is the zero scheme of a section of the globally generated bundle $\mathcal O(h+a)^{\oplus m}\oplus\mathcal O(h+b)^{\oplus(m-1)}\oplus\mathcal O(h)\oplus\mathcal O(a+b)^{\oplus(N-2)}$, whose rank $2m+N-2$ equals $\dim P$. For such a section the number of isolated zeros is at most the degree of the top Chern class \cite[Ch.~12]{Ful98}, which is
\[
\bigl[h^{N}a^{m-1}b^{m-1}\bigr](h+a)^m(h+b)^{m-1}\,h\,(a+b)^{N-2}\;=\;\Bbez(N,m).
\]
Hence $\delta_{\mathrm{top}}(X)\le\Bbez(N,m)$.

\emph{Conclusion.} By \eqref{eq:bez-upper}, $\delta_{\mathrm{top}}(X)\le(4e\,m/(N-1))^{N-1}$ for every size $m$ of a representation of $f$; take $m=\dc(f)$ and use $\dc(f)\le\dc(P)$ (Lemma~\ref{lem:monotone}). For smooth $X$, $\delta_{\mathrm{top}}(X)=d(d-1)^{N-2}\ge(d-1)^{N-1}$.
\end{proof}

In the language of Section~\ref{sec:nogo}, this proof reads the invariant $\delta_{\mathrm{top}}(V(f))=\delta_1(\Con V(G))$ through the slot-$1$ corank-one system in $\PP^N$, and bounds the count by $\Bbez(N,m)$, which by Theorem~\ref{thm:T} and Proposition~\ref{prop:bez} dominates the exact generic count $T_1(N+1,m)=T(N,m)$.

\subsection{The border and symmetric models}

\begin{theorem}[Smooth slices \cite{She26a}]\label{thm:smooth-slice}
Let $P$ be a polynomial that has a linear restriction equal to a homogeneous form $f(y_1,\dots,y_N)$ of degree $d\ge2$, where $N\ge3$ and $V(f)\subset\PP^{N-1}$ is smooth. Then
\begin{gather*}
\dcb(P)\;\ge\;\frac{N-1}{4e}\bigl(d(d-1)^{N-2}\bigr)^{1/(N-1)}\;\ge\;\frac{(N-1)(d-1)}{4e},\\
\sdcb(P)\;\ge\;\frac{N-1}{2e}\bigl(d(d-1)^{N-2}\bigr)^{1/(N-1)}\;\ge\;\frac{(N-1)(d-1)}{2e}.
\end{gather*}
\end{theorem}

\begin{remark}[What is used from \cite{She26a}]\label{rem:source}
By Lemma~\ref{lem:monotone} it suffices to treat $P=f$. The paper \cite{She26a} states and proves the theorem for the power sum $y_1^n+\dots+y_n^n$, so that $N=d=n$, and its proof in that case is formally verified in Lean~4. The proof compares the same two quantities as Theorem~\ref{thm:polar-exact}: on the target side the dual degree of $V(f)$, which equals $d(d-1)^{N-2}$ for every smooth hypersurface of degree $d$ in $\PP^{N-1}$, and on the determinant side the count $\Bbez(N,m)$, which does not depend on the target. To this it adds a specialization argument showing that the comparison survives limits of determinants. That argument uses only that $f$ is homogeneous and $V(f)$ is smooth, once the number of variables $N$ and the degree $d$, which coincide for the power sum, are kept separate. In the symmetric models the kernel and cokernel coincide, the incidence lives in $\PP^{N-1}\times\PP^{m-1}$, and the corresponding count is at most $(2e\,m/(N-1))^{N-1}$. For $N\ne d$ the border and symmetric statements therefore rest on the written argument of \cite{She26a}; the exact statement is Theorem~\ref{thm:polar-exact}.
\end{remark}

The case used most often is the Fermat polynomial $\Phi_{N,d}:=y_1^d+\dots+y_N^d$, whose partial derivatives $dy_i^{d-1}$ vanish simultaneously only at the origin.

\begin{corollary}\label{cor:fermat}
If $P$ has a linear restriction equal to $\Phi_{N,d}$ with $N\ge3$ and $d\ge2$, then $P$ satisfies the bounds of Theorem~\ref{thm:smooth-slice}.
\end{corollary}

\subsection{The codimension principle}

For a nonzero homogeneous polynomial $P$ in $M$ variables, let $c(P)=M$ if $V(P)$ is smooth and $c(P)=\codim_{\PP^{M-1}}\Sing V(P)$ otherwise.

\begin{theorem}[Codimension principle]\label{thm:codim}
Let $P\in\C[z_1,\dots,z_M]$ be homogeneous and squarefree of degree $d\ge2$, and suppose $c=c(P)\ge3$. Then
\[
\dcb(P)\;\ge\;\frac{c-1}{4e}\bigl(d(d-1)^{c-2}\bigr)^{1/(c-1)}\;\ge\;\frac{(c-1)(d-1)}{4e},
\]
and the same holds for $\sdcb$ with $2e$ in place of $4e$.
\end{theorem}

\begin{proof}
If $V(P)$ is smooth, apply Theorem~\ref{thm:smooth-slice} to $P$. Otherwise $\dim\Sing V(P)=M-1-c$. A generic linear subspace $\Lambda\subset\PP^{M-1}$ of dimension $c-1$ satisfies $\dim\Lambda+\dim\Sing V(P)=M-2<M-1$, so it misses $\Sing V(P)$. Since $P$ is squarefree, $V(P)\setminus\Sing V(P)$ is smooth, and by Bertini's theorem its intersection with a generic $\Lambda$ is smooth. So, in coordinates on $\Lambda\cong\PP^{c-1}$, the restriction $P|_\Lambda$ is a form of degree $d$ in $c$ variables with smooth zero set; it is nonzero because a generic $\Lambda$ is not contained in $V(P)$. Apply Theorem~\ref{thm:smooth-slice} with $N=c$.
\end{proof}

Read in reverse, the theorem says that polynomials that are easy for border determinants define very singular hypersurfaces: if $\dcb(P)<(k-1)(d-1)/(4e)$ for some $k\ge3$, then $\codim\Sing V(P)<k$.

\begin{remark}[The determinant]\label{rem:det}
The singular locus of $\{\det_m=0\}$ is the set of matrices of rank at most $m-2$, of codimension $4$. Theorem~\ref{thm:codim} therefore gives only $\dcb(\det_m)\ge\frac3{4e}(m(m-1)^2)^{1/3}$, a bound linear in $m$, consistent with $\dc(\det_m)=m$.
\end{remark}

For the permanent, generic slices give at most a quadratic bound.

\begin{proposition}[Generic slices of the permanent]\label{prop:perm-codim}
Let $n\ge2$ be the size of the permanent. Theorem~\ref{thm:codim} applied to $\per_n$ gives a lower bound of at most
\[
\frac{2n-1}{4e}\bigl(n(n-1)^{2n-2}\bigr)^{1/(2n-1)}=\Bigl(\frac1{2e}+o(1)\Bigr)n^2 .
\]
\end{proposition}

\begin{proof}
Matrices whose first two rows vanish form a linear subspace of $\PP^{n^2-1}$ of codimension $2n$. Every first partial derivative of $\per_n$ is the permanent of an $(n-1)\times(n-1)$ submatrix, which still has a zero row, so all of them vanish on this subspace and $c(\per_n)\le2n$. The bound of Theorem~\ref{thm:codim} is increasing in $c$, and $d=n$.
\end{proof}

The known bound $\dcb(\per_n)\ge n^2/2$ \cite{LMR13} is larger, so this is a limitation of generic slices rather than of the method: special restrictions can do better, and the cubic bound of \cite{OAI26} reaches the ceiling of Corollary~\ref{cor:perm} in order.

\section{Homomorphism polynomials}\label{sec:hom}

\subsection{Definitions and completeness}

Let $F$ be a finite bipartite multigraph with sides $L$ and $R$. For $q\ge1$ the \emph{homomorphism polynomial} of $F$ is
\[
\hom_{F,q}(X)=\sum_{\alpha:L\to[q]}\ \sum_{\beta:R\to[q]}\ \prod_{uv\in E(F)} x_{\alpha(u),\beta(v)}
\ \in\ \C[x_{ij}: 1\le i,j\le q],
\]
where $u\in L$, $v\in R$, and parallel edges are counted with multiplicity. It has $q^2$ variables and degree $|E(F)|$. If the variables are set to the entries of the biadjacency matrix of a bipartite graph $G$ with $q$ vertices on each side, $\hom_{F,q}$ counts homomorphisms from $F$ to $G$; this is the normalization of \cite{DPS26}.

A sequence $(F_q)$ of multigraphs is a \emph{$p$-family} if $|V(F_q)|+|E(F_q)|$ is bounded by a polynomial in $q$. We write $\tw$, $\pw$ and $\td$ for treewidth, pathwidth and treedepth. Constant-depth $c$-reductions are a notion of reduction between polynomial families that is weaker than $p$-projections; see \cite{DPS26} for the definition.

\begin{theorem}[Dwivedi, Pago and Seppelt {\cite[Theorem~D]{DPS26}}]\label{thm:dps}
Let $(F_q)$ be a $p$-family of bipartite multigraphs and let $\varepsilon>0$.
\begin{enumerate}[label=(\roman*)]
\item If $\tw(F_q)\ge q^{\varepsilon}$ for all $q\ge2$, then $(\hom_{F_q,q})$ is $\VNP$-complete.
\item If $\tw(F_q)=O(1)$ and $\pw(F_q)\ge\varepsilon\log q$ for all $q\ge2$, then $(\hom_{F_q,q})$ is $\VP$-complete.
\item If $\pw(F_q)=O(1)$ and $\td(F_q)\ge\varepsilon\log q$ for all $q\ge2$, then $(\hom_{F_q,q})$ is $\VBP$-complete.
\end{enumerate}
In all cases hardness holds under constant-depth $c$-reductions over every field of characteristic zero.
\end{theorem}

Complete families of homomorphism polynomials were first found for specific patterns by Durand, Mahajan, Malod, de Rugy-Altherre and Saurabh \cite{DMMRS14} and by Mahajan and Saurabh \cite{MS18}; Theorem~\ref{thm:dps} gives general criteria. The lower bounds below do not depend on completeness, which only explains why these families are of interest.

\subsection{The diagonal slice}

\begin{lemma}\label{lem:diag-slice}
Let $F$ be a connected bipartite multigraph. Substituting $x_{ii}=y_i$ for $1\le i\le q$ and $x_{ij}=0$ for $i\ne j$ gives
\[
\hom_{F,q}(\mathrm{diag}(y_1,\dots,y_q))=\sum_{i=1}^q y_i^{|E(F)|}.
\]
\end{lemma}

\begin{proof}
The term indexed by $(\alpha,\beta)$ survives exactly when $\alpha(u)=\beta(v)$ for every edge $uv$. Since $F$ is connected, this happens exactly when all vertices receive the same label $i\in[q]$, and each of these $q$ labellings contributes $y_i^{|E(F)|}$.
\end{proof}

\begin{theorem}\label{thm:hom-lower}
Let $F$ be a connected bipartite multigraph with $d=|E(F)|\ge2$ edges, and let $q\ge3$. Then
\[
\dcb(\hom_{F,q})\;\ge\;\frac{q-1}{4e}\bigl(d(d-1)^{q-2}\bigr)^{1/(q-1)}\;\ge\;\frac{(q-1)(d-1)}{4e},
\]
and $\sdcb(\hom_{F,q})$ satisfies the same bounds with $2e$ in place of $4e$.
\end{theorem}

\begin{proof}
By Lemma~\ref{lem:diag-slice}, $\hom_{F,q}$ has a linear restriction equal to $\Phi_{q,d}$. Apply Corollary~\ref{cor:fermat} with $N=q$.
\end{proof}

\begin{remark}[How strong is this?]\label{rem:normalized}
Since $\hom_{F,q}$ has $M=q^2$ variables, the theorem gives $\dcb(\hom_{F,q})\ge(\frac1{4e}-o(1))\deg(\hom_{F,q})\sqrt M$ as $q,d\to\infty$, an improvement over the degree bound of Lemma~\ref{lem:monotone} by a factor of order $\sqrt M$. It is also all that the method can extract from this slice: the slice $\Phi_{q,d}$ has $q$ variables and degree $d$, so by Theorem~\ref{thm:nogo} no argument of the kind considered in Section~\ref{sec:nogo} certifies more than $O(dq)$ from it. Beating $\Theta(dq)$ for $\hom_{F,q}$ with this method requires smooth restrictions in more than $q$ variables, or more of the singular geometry of $V(\hom_{F,q})$ (Section~\ref{sec:Hq}).
\end{remark}

\subsection{Complete families}

\paragraph{Paths: a tight bound for a $\VBP$-complete family.}
Let $P_L$ be the path with $L$ edges, with its natural bipartition, and set $W_q:=\hom_{P_{q^2},q}$, a polynomial of degree $q^2$ in $q^2$ variables.

\begin{theorem}\label{thm:vbp}
The family $(W_q)$ is $\VBP$-complete under constant-depth $c$-reductions over fields of characteristic zero. For $q\ge3$,
\[
\dcb(W_q)\;\ge\;\frac{(q-1)(q^2-1)}{4e},\qquad \sdcb(W_q)\;\ge\;\frac{(q-1)(q^2-1)}{2e},
\]
and $\dc(W_q)=O(q^3)$. In particular $\dcb(W_q)=\Theta(q^3)$.
\end{theorem}

\begin{proof}
The path $P_{q^2}$ has pathwidth $1$ and treedepth $\Theta(\log q)$, so Theorem~\ref{thm:dps}(iii) gives completeness. The lower bounds are Theorem~\ref{thm:hom-lower} with $d=q^2$. For the upper bound, list the vertices of the path in order as $v_0,\dots,v_{q^2}$ and build a layered algebraic branching program with a source, a sink, and $q^2+1$ layers of $q$ vertices, where vertex $a$ of layer $t$ stands for ``$v_t$ has label $a$''. Join the source to every vertex of layer $0$ and every vertex of the last layer to the sink with weight $1$, and vertex $a$ of layer $t$ to vertex $b$ of layer $t+1$ with weight $x_{a,b}$ if $v_t$ is on the left side and $x_{b,a}$ otherwise. Source--sink paths correspond to labellings, with the right weights, so the program computes $W_q$. It has $O(q^3)$ vertices, and the standard conversion of algebraic branching programs into determinants \cite{Val79,Bur00} gives $\dc(W_q)=O(q^3)$.
\end{proof}

\paragraph{Complete binary trees: a $\VP$-complete family.}
Let $h_q=\lfloor 2\log_2 q\rfloor$ and let $B_q$ be the complete binary tree of height $h_q$. It has treewidth $1$, pathwidth $\Theta(\log q)$ and $D_q:=|E(B_q)|=2^{h_q+1}-2$ edges, and since $q^2/2<2^{h_q}\le q^2$ we have $q^2-2<D_q\le 2q^2-2$. Set $T_q:=\hom_{B_q,q}$.

\begin{theorem}\label{thm:vp}
The family $(T_q)$ is $\VP$-complete under constant-depth $c$-reductions over fields of characteristic zero, and for $q\ge3$
\[
\dcb(T_q)\;\ge\;\frac{(q-1)(D_q-1)}{4e}\;\ge\;\frac{(q-1)(q^2-2)}{4e}=\Bigl(\frac{1}{4e}-o(1)\Bigr)q^3,
\]
with $2e$ in place of $4e$ for $\sdcb(T_q)$.
\end{theorem}

\begin{proof}
Completeness is Theorem~\ref{thm:dps}(ii), and the bounds are Theorem~\ref{thm:hom-lower}, using $D_q-1\ge q^2-2$.
\end{proof}

\paragraph{Grids: a planar $\VNP$-complete family.}
Let $G_q$ be the $q\times q$ grid graph: connected, bipartite, planar, of maximum degree $4$ and treewidth $q$, with $2q(q-1)$ edges. Set $R_q:=\hom_{G_q,q}$.

\begin{theorem}\label{thm:vnp}
The family $(R_q)$ is $\VNP$-complete under constant-depth $c$-reductions over fields of characteristic zero, and for $q\ge3$
\[
\dcb(R_q)\;\ge\;\frac{(q-1)(2q(q-1)-1)}{4e}=\Bigl(\frac{1}{2e}-o(1)\Bigr)q^3,
\qquad
\sdcb(R_q)\;\ge\;\Bigl(\frac{1}{e}-o(1)\Bigr)q^3.
\]
\end{theorem}

\begin{proof}
Completeness is Theorem~\ref{thm:dps}(i), and the bounds are Theorem~\ref{thm:hom-lower} with $d=2q(q-1)$.
\end{proof}

\paragraph{A family complete under $p$-projections.}
The families above are complete under constant-depth $c$-reductions. Adding loops to a directed construction of Chaugule, Limaye and Varre \cite{CLV21} gives a family that is $\VNP$-complete under the stricter $p$-projections and still has a large Fermat slice. For directed graphs $G=(V_G,E_G)$ and $H=(V_H,E_H)$, the directed homomorphism polynomial is
\[
\dhom_{G,H}(Y)=\sum_{\phi:G\to H}\ \prod_{(u,v)\in E_G} Y_{\phi(u),\phi(v)},
\]
summed over homomorphisms $\phi$, with one variable $Y_{a,b}$ for each edge $(a,b)$ of $H$. Let $H^\circ$ be $H$ with a loop added at every vertex that has none.

\begin{lemma}\label{lem:loop}
If the underlying undirected graph of $G$ is connected, then $\dhom_{G,H^\circ}$ has a linear restriction equal to $\sum_{w\in V_H}y_w^{|E_G|}$. Consequently, if $|V_H|\ge3$ and $|E_G|\ge2$, then $\dcb(\dhom_{G,H^\circ})\ge(|V_H|-1)(|E_G|-1)/(4e)$, with $2e$ in place of $4e$ for $\sdcb$.
\end{lemma}

\begin{proof}
Set each loop variable $Y_{w,w}$ to $y_w$ and all other variables to $0$. A homomorphism survives only if it maps every edge to a loop, so by connectivity it is constant, and the constant map to $w$ contributes $y_w^{|E_G|}$. Apply Corollary~\ref{cor:fermat}.
\end{proof}

\begin{lemma}\label{lem:preserve}
Let $(\dhom_{G_n,H_n})$ be a $p$-family that is $\VNP$-hard under $p$-projections. Then the family $(\dhom_{G_n,H_n^\circ})$ is $\VNP$-complete under $p$-projections.
\end{lemma}

\begin{proof}
Setting the added loop variables to zero recovers $\dhom_{G_n,H_n}$, so the original family is a $p$-projection of the new one; membership in $\VNP$ follows by summing over the vertex map $\phi$.
\end{proof}

The $\VNP$-complete directed family of \cite{CLV21} uses a source graph with vertices $x$, $y$, $u_1,\dots,u_n$ and $v_1,\dots,v_n$, and with edges
\begin{gather*}
  (x,u_i),\quad (y,v_i),\quad (u_i,v_i),\quad (u_i,y)\qquad (1\le i\le n),\\
  (x,y),\qquad (u_i,u_j)\ \ (i\ne j),\qquad (v_i,v_j)\ \ (i<j).
\end{gather*}
Its underlying undirected graph is connected, it has $D_n=4n+1+n(n-1)+\frac{n(n-1)}{2}=\frac{3n^2+5n+2}{2}$ edges, and the target graph of the construction has $2n+2$ vertices. Let $L_n$ be the corresponding directed homomorphism polynomial with loops added to the target.

\begin{theorem}\label{thm:clv}
The family $(L_n)$ is $\VNP$-complete under $p$-projections, and
\[
\dcb(L_n)\;\ge\;\frac{(2n+1)(3n^2+5n)}{8e}=\Bigl(\frac{3}{4e}-o(1)\Bigr)n^3,
\qquad
\sdcb(L_n)\;\ge\;\Bigl(\frac{3}{2e}-o(1)\Bigr)n^3.
\]
\end{theorem}

\begin{proof}
Completeness follows from Lemma~\ref{lem:preserve} and \cite{CLV21}; the bounds follow from Lemma~\ref{lem:loop} with $|V_H|=2n+2$ and $|E_G|=D_n$.
\end{proof}

\subsection{The dense family \texorpdfstring{$H_q$}{Hq}}\label{sec:Hq}

In this subsection we study
\[
H_q(X):=\hom_{K_{q,q},q}(X)=\sum_{\alpha,\beta:[q]\to[q]}\ \prod_{i,j=1}^q x_{\alpha(i),\beta(j)},
\]
which has degree $q^2$ in $q^2$ variables. Since $\tw(K_{q,q})=q$, the family $(H_q)$ is $\VNP$-complete by Theorem~\ref{thm:dps}(i), and Theorem~\ref{thm:hom-lower} gives $\dcb(H_q)\ge(q-1)(q^2-1)/(4e)$ for $q\ge3$. The ceiling of Theorem~\ref{thm:nogo} for $H_q$ is $O(q^4)$, so there is room for a factor $q$. For $\alpha,\beta$ let $r_k=|\alpha^{-1}(k)|$ and $s_l=|\beta^{-1}(l)|$; the term indexed by $(\alpha,\beta)$ is the monomial $\prod_{k,l}x_{kl}^{r_ks_l}$.

\begin{lemma}\label{lem:moment}
$H_q(X)=\sum_{s}\binom{q}{s_1,\dots,s_q}\Bigl(\sum_{k=1}^q\prod_{l=1}^q x_{kl}^{s_l}\Bigr)^{q}$, where the sum is over all $s\in\Z_{\ge0}^q$ with $s_1+\dots+s_q=q$.
\end{lemma}

\begin{proof}
Group the maps $\beta$ by $s$; there are $\binom{q}{s_1,\dots,s_q}$ of them for each $s$. For fixed $\beta$ the sum over $\alpha$ factors over $i\in[q]$, and each factor is $\sum_k\prod_l x_{kl}^{s_l}$.
\end{proof}

\paragraph{Rank-one matrices.}
Let $p_q(u)=u_1^q+\dots+u_q^q$, and let $\Seg\subseteq\PP^{q^2-1}$ be the variety of rank-one matrices $uw^{\mathsf T}$.

\begin{theorem}\label{thm:rankone}
For all $q\ge2$, all $u,w\in\C^q$ and all $k,l$,
\[
H_q(uw^{\mathsf T})=\bigl(p_q(u)\,p_q(w)\bigr)^q,
\qquad
\frac{\partial H_q}{\partial x_{kl}}(uw^{\mathsf T})=q^2\,u_k^{q-1}w_l^{q-1}\bigl(p_q(u)\,p_q(w)\bigr)^{q-1}.
\]
Consequently $\Seg\cap V(H_q)=\{uw^{\mathsf T}: p_q(u)p_q(w)=0\}\subseteq\Sing V(H_q)$.
\end{theorem}

\begin{proof}
At $X=uw^{\mathsf T}$ the term indexed by $(\alpha,\beta)$ equals $\prod_{i}u_{\alpha(i)}^q\prod_j w_{\beta(j)}^q$, and summing gives the first formula. For the derivative, assume first that all $u_k,w_l$ are nonzero. Differentiating the term indexed by $(\alpha,\beta)$ with respect to $x_{kl}$ at $uw^{\mathsf T}$ multiplies its value by $r_ks_l/(u_kw_l)$, and $\sum_\alpha r_k\prod_i u_{\alpha(i)}^q=q\,u_k^q\,p_q(u)^{q-1}$, as one sees by applying $z_k\,\partial/\partial z_k$ to $(z_1+\dots+z_q)^q$ and substituting $z_i=u_i^q$. The same holds for $w$, and multiplying gives the formula, which then holds everywhere by continuity. If $p_q(u)p_q(w)=0$, the value and all first derivatives vanish because $q-1\ge1$.
\end{proof}

\begin{corollary}\label{cor:uppercodim}
For $q\ge2$, $\codim\Sing V(H_q)\le(q-1)^2+1$.
\end{corollary}

\begin{proof}
The image of $\{p_q(u)=0\}\times\PP^{q-1}$ under the Segre embedding has dimension $2q-3$ and lies in $\Sing V(H_q)$, so the codimension in $\PP^{q^2-1}$ is at most $(q^2-1)-(2q-3)$.
\end{proof}

\paragraph{The diagonal and irreducibility.}
Let $\Delta\cong\PP^{q-1}$ be the linear subspace of diagonal matrices.

\begin{lemma}\label{lem:diag-avoid}
For every $q\ge2$, $\Delta\cap\Sing V(H_q)=\emptyset$. Consequently $\codim\Sing V(H_q)\ge q$.
\end{lemma}

\begin{proof}
By Lemma~\ref{lem:diag-slice}, the restriction of $H_q$ to $\Delta$ is $\sum_iy_i^{q^2}$. A singular point of $V(H_q)$ on $\Delta$ would, by the chain rule, be a singular point of the smooth hypersurface $V(\sum_iy_i^{q^2})$. Since disjoint subvarieties of $\PP^{q^2-1}$ have dimensions summing to less than $q^2-1$, $\dim\Sing V(H_q)\le q^2-q-1$.
\end{proof}

\begin{theorem}\label{thm:irreducible}
For every $q\ge3$ the polynomial $H_q$ is irreducible. For $q=2$,
\[
H_2=s^2-2\det(X)^2=\bigl(s-\sqrt2\det X\bigr)\bigl(s+\sqrt2\det X\bigr),\qquad s=\sum_{i,j}x_{ij}^2.
\]
\end{theorem}

\begin{proof}
If $H_q$ had a repeated factor $A$, then $V(A)\subseteq\Sing V(H_q)$; but $\dim V(A)+\dim\Delta=(q^2-2)+(q-1)\ge q^2-1$, so $V(A)$ meets $\Delta$, contradicting Lemma~\ref{lem:diag-avoid}. If $H_q=PQ$ with $P,Q$ nonconstant and coprime, every point of $V(P)\cap V(Q)$ is singular on $V(H_q)$, and every component of $V(P)\cap V(Q)$ has dimension at least $q^2-3$; for $q\ge3$, $(q^2-3)+(q-1)\ge q^2-1$, so again $V(P)\cap V(Q)$ meets $\Delta$. The formula for $q=2$ follows by expanding the definition (Appendix~\ref{app:comp}).
\end{proof}

\begin{corollary}\label{cor:sandwich}
For $q\ge3$, $q\le\codim\Sing V(H_q)\le(q-1)^2+1$.
\end{corollary}

Since $H_q$ is squarefree for $q\ge3$, Theorem~\ref{thm:codim} applies; with $c(H_q)\ge q$ it recovers the diagonal bound, and any improvement of the lower bound on $c(H_q)$ improves the bound on $\dcb(H_q)$.

\paragraph{Coordinate slices.}
A natural way to improve on the diagonal slice is to keep more coordinates. For $E\subseteq[q]\times[q]$ let $H_{q,E}$ be obtained from $H_q$ by setting $x_{ij}=0$ for $(i,j)\notin E$, a form of degree $q^2$ in $|E|$ variables. Call $E$ a \emph{matching} if no two positions in $E$ share a row or a column.

\begin{theorem}\label{thm:coord-max}
Let $q\ge2$ and $|E|\ge2$. The hypersurface $V(H_{q,E})\subseteq\PP^{|E|-1}$ is smooth if and only if $E$ is a matching, and in that case $H_{q,E}=\sum_{(i,j)\in E}x_{ij}^{q^2}$.
\end{theorem}

\begin{proof}
The support of the term indexed by $(\alpha,\beta)$ is the rectangle $\mathrm{supp}(r)\times\mathrm{supp}(s)$, and the term survives in $H_{q,E}$ exactly when this rectangle lies in $E$. If $E$ is a matching, the only such rectangles are single positions $\{(i,j)\}$, coming from $\alpha\equiv i$ and $\beta\equiv j$, which proves the formula; Fermat hypersurfaces are smooth.

Suppose $E$ is not a matching; by symmetry, $E$ contains $e=(a,b)$ and $f=(a,c)$ with $b\ne c$. The rectangles inside $\{e,f\}$ are $\{a\}\times S$ with $\emptyset\ne S\subseteq\{b,c\}$, coming from $\alpha\equiv a$ and $\beta:[q]\to\{b,c\}$, so the part of $H_{q,E}$ involving only $x_e,x_f$ is $\sum_{\beta}\prod_j x_{a,\beta(j)}^{q}=(x_e^q+x_f^q)^q$. Choose nonzero $s,t$ with $s^q+t^q=0$ and let $p$ have $x_e=s$, $x_f=t$ and all other coordinates $0$. Then $H_{q,E}(p)=0$, and the derivatives with respect to $x_e,x_f$ vanish at $p$ because $q\ge2$. For $g=(i,j)\in E\setminus\{e,f\}$, the value of $\partial H_{q,E}/\partial x_g$ at $p$ comes only from terms in which $x_g$ has exponent exactly $1$ and every other variable is $x_e$ or $x_f$. The exponent of $x_g$ is $r_is_j$, so $r_i=s_j=1$; since $\sum r_k=\sum s_l=q\ge2$, there are $i'\ne i$ and $j'\ne j$ with $r_{i'},s_{j'}\ge1$, and then the support also contains $(i',j)$, $(i,j')$ and $(i',j')$, three positions distinct from each other and from $g$, which cannot all lie in $\{e,f\}$. So $p$ is a singular point.
\end{proof}

A matching has at most $q$ positions, so smooth coordinate slices of $H_q$ give at best the diagonal bound; a better smooth slice must use non-coordinate linear forms. For example, the positions $(i,i)$ and $(i,i+1)$, $i\in[q]$ (indices modulo $q$), share rows, so the corresponding slice is singular at the point with $x_{rr}=a$, $x_{r,r+1}=b$, $a^q+b^q=0$, $ab\ne0$, and all other coordinates $0$.

\paragraph{Further singular points and a conjecture.}
The rank-one singular points are not all of $\Sing V(H_q)$; already for $q=3$ there are singular points of rank two.

\begin{proposition}\label{prop:h3-ranktwo}
On the matrices $X=\begin{pmatrix}0&0&a\\0&0&b\\c&g&0\end{pmatrix}$, the partial derivatives of $H_3$ with respect to $x_{11},x_{12},x_{21},x_{22},x_{33}$ vanish, and
\[
\frac{\partial H_3}{\partial x_{13}}=9a^2(a^3+b^3)^2,\quad
\frac{\partial H_3}{\partial x_{23}}=9b^2(a^3+b^3)^2,\quad
\frac{\partial H_3}{\partial x_{31}}=9c^2(c^3+g^3)^2,\quad
\frac{\partial H_3}{\partial x_{32}}=9g^2(c^3+g^3)^2,
\]
\[
H_3(X)=(a^3+b^3+c^3+g^3)\bigl(a^6+2a^3b^3-a^3c^3-a^3g^3+b^6-b^3c^3-b^3g^3+c^6+2c^3g^3+g^6\bigr).
\]
Hence every nonzero such $X$ with $a^3+b^3=c^3+g^3=0$, for instance $a=c=1$, $b=g=-1$, is a singular point of $V(H_3)$ of rank two.
\end{proposition}

\begin{proof}
Substitute $X$ into Lemma~\ref{lem:moment} for $q=3$ and into its derivatives and expand; this finite computation was checked with a computer algebra system (Appendix~\ref{app:comp}).
\end{proof}

This rank-two family has dimension one in $\PP^8$ for each pattern of positions, while the rank-one singular locus has dimension $2q-3=3$, so it does not affect the codimension bound.

\begin{conjecture}\label{conj:topdim}
For every $q\ge3$, $\dim\Sing V(H_q)=2q-3$; equivalently, $\codim\Sing V(H_q)=(q-1)^2+1$.
\end{conjecture}

\begin{corollary}\label{cor:conj-conseq}
If Conjecture~\ref{conj:topdim} holds, then $\dcb(H_q)\ge(q-1)^2(q^2-1)/(4e)=\Omega(q^4)$. By Theorem~\ref{thm:nogo} this is, up to a constant factor, the largest bound that the method can give for $H_q$.
\end{corollary}

\begin{proof}
Apply Theorem~\ref{thm:codim} with $d=q^2$ and $c=(q-1)^2+1$; it applies because $H_q$ is squarefree for $q\ge3$. The ceiling for degree $q^2$ in $q^2$ variables is $O(q^4)$.
\end{proof}

\begin{remark}[Evidence]\label{rem:evidence}
The first open case is $q=3$, where the conjecture says that $\Sing V(H_3)$ has dimension $3$ in $\PP^8$; deciding it means computing the dimension of the ideal generated by nine forms of degree $8$ in nine variables, which we have not completed. Appendix~\ref{app:comp} reports exact counts of the points of $\Sing V(H_3)$ over $\FF_5$ and $\FF_7$, and shows that one randomly chosen $\PP^4\subseteq\PP^8$ contains no point of $\Sing V(H_3)$ over $\FF_5$ or $\FF_{25}$. A subvariety of $\PP^8$ of dimension at least $4$ meets every $\PP^4$ over the algebraic closure, so the second check is consistent with the conjecture; neither check determines the dimension over $\C$.
\end{remark}

\section{The invariant ceiling}\label{sec:ceiling}

This section shows that every characteristic-cycle invariant of a hypersurface of degree $d$ is at most about $8(d-1)^{N-1}$. The argument has three steps: bound the Euler characteristics of generic linear sections (Section~\ref{sec:envelope}), convert them into multidegrees by a fixed linear map with small coefficients (Section~\ref{sec:transform}), and use positivity to bound each coefficient by the total (Section~\ref{sec:ceilingthm}). Here $f$ is an arbitrary nonzero homogeneous polynomial of degree $d\ge2$ in $N$ variables, $X=V(f)_{\mathrm{red}}\subset\PP^n$ with $n=N-1$, and no hypothesis is placed on $\Sing X$. The degree $d$ always refers to the given equation, which dominates the degree of the reduction.

\subsection{The Euler-characteristic envelope}\label{sec:envelope}

\begin{lemma}[Global Milnor fibre as a cover]\label{lem:cover}
Let $g$ be a nonzero homogeneous polynomial of degree $d\ge1$ in $k+1$ variables, $U=\PP^k\setminus V(g)$ and $F=\{g=1\}\subset\C^{k+1}$. Then $F$ is a smooth affine variety of dimension $k$, the map $F\to U$, $x\mapsto[x]$, is a $d$-sheeted covering, and $\chi(F)=d\,\chi(U)$.
\end{lemma}

\begin{proof}
By Euler's relation $\sum_ix_i\partial_ig=dg$, the gradient of $g$ does not vanish on $F$. The fibre over $[x]$ is $\{tx:\,t^dg(x)=1\}$, which has $d$ points, and Euler characteristics are multiplicative in finite coverings.
\end{proof}

\begin{lemma}[Critical points]\label{lem:crit}
Let $g,F$ be as in Lemma~\ref{lem:cover} with $d\ge2$. For generic linear functions $\ell,\ell_1,\dots,\ell_{k-j}$ on $\C^{k+1}$ and generic $t_1,\dots,t_{k-j}\in\C$, the restriction of $\ell$ to the smooth $j$-dimensional affine variety $F_j:=F\cap\{\ell_1=t_1,\dots,\ell_{k-j}=t_{k-j}\}$ is a Morse function with at most $d(d-1)^{\,j}$ critical points, for every $0\le j\le k$.
\end{lemma}

\begin{proof}
For generic data $\ell|_{F_j}$ is Morse with finitely many critical points. Choose coordinates in which $\ell_1,\dots,\ell_{k-j},\ell$ span the first $k-j+1$ coordinate functions. By Lagrange, $x\in F_j$ is critical exactly when
\[
g(x)=1,\qquad \ell_s(x)=t_s\ (s\le k-j),\qquad \partial_ig(x)=0\ \ (k-j+2\le i\le k+1),
\]
a system of $k+1$ equations of degrees $d,1,\dots,1,d-1,\dots,d-1$ whose solutions with $\nabla g=0$ are excluded by Euler's relation. Morse critical points are isolated solutions. After homogenizing ($g=x_0^d$, $\ell_s=t_sx_0$), an isolated affine solution stays an isolated point of the projective intersection, and by the refined B\'ezout theorem \cite[$\S$12.3]{Ful98} the number of isolated points of an intersection of hypersurfaces is at most the product of their degrees, $d(d-1)^j$.
\end{proof}

The passage from critical points to Euler characteristics uses the Lefschetz-pencil formula of Seade, Tib\u{a}r and Verjovsky, in the special case of smooth varieties.

\begin{theorem}[{\cite[Lemma~3.1, Theorem~3.4 and formula~(2), with $\S$2.1(a)]{STV05}}]\label{thm:STV}
Let $Y\subset\C^M$ be a smooth closed affine variety of pure dimension $k\ge1$. For linear functions $\ell$ in a dense open set, $\ell|_Y$ has finitely many critical points, all Morse, and for generic $t\in\C$
\[
\chi(Y)\;=\;\chi\bigl(Y\cap\{\ell=t\}\bigr)\;+\;(-1)^{k}\,\#\mathrm{crit}(\ell|_Y).
\]
\end{theorem}

In \cite{STV05} the formula is stated for the global Euler obstruction of an arbitrary affine variety; for smooth $Y$ the global Euler obstruction is $\chi(Y)$ \cite[$\S$2.1(a)]{STV05}, the generic slice is again smooth, and the critical-point counts on iterated generic slices do not depend on the generic choices \cite[$\S$3.1]{STV05}.

\begin{proposition}[Envelope]\label{prop:envelope}
Let $f$ be nonzero homogeneous of degree $d\ge2$ in $N$ variables and $X=V(f)\subset\PP^{n}$. For every $0\le k\le n$ and generic linear $L_k\subset\PP^n$ of dimension $k$,
\[
\bigl|\chi(X\cap L_k)\bigr|\;\le\;E(d,k)\;:=\;2(d-1)^k+2k+2 .
\]
\end{proposition}

\begin{proof}
$X\cap L_k=V(g)$ for a nonzero form $g$ of degree $d$ on $L_k\cong\PP^k$. With $F,U$ as in Lemma~\ref{lem:cover}, apply Theorem~\ref{thm:STV} to $F$ and then repeatedly to its generic slices, down to dimension $0$, where the slice has at most $d$ points; bounding each critical-point count by Lemma~\ref{lem:crit} gives $|\chi(F)|\le\sum_{j=0}^{k}d(d-1)^{j}$. For $d\ge3$ this is at most $d(d-1)^k\frac{d-1}{d-2}\le2d(d-1)^k$, so $|\chi(U)|\le2(d-1)^k$ and $|\chi(V(g))|\le\chi(\PP^k)+|\chi(U)|\le(k+1)+2(d-1)^k$. For $d=2$ the sum is $2(k+1)$, so $|\chi(V(g))|\le2k+2$.
\end{proof}

\subsection{From slices to multidegrees}\label{sec:transform}

\begin{theorem}[Second-difference transform]\label{thm:transform}
Let $\varphi$ be a constructible function on $\PP^n$, with $\chi_k=\chi_k(\varphi)$ and $\chi_{-1}=\chi_{-2}=0$. Then
\[
\Delta_i(\varphi)\;=\;(-1)^i\bigl(\chi_{n-i}-2\chi_{n-i-1}+\chi_{n-i-2}\bigr)\quad(0\le i\le n-1),
\qquad n_0(\varphi)=(-1)^n\chi_0(\varphi).
\]
In particular each $\Delta_i$ is a combination of three slice values with coefficients $1,-2,1$, independently of $n$.
\end{theorem}

\begin{proof}
By the projective form of the index theorem (Section~\ref{sec:cc}), there is an integer matrix $A$ of size $n+1$, depending only on $n$, such that for every $\varphi$
\[
(\chi_0,\dots,\chi_n)^{\mathsf T}=A\,(n_0,\Delta_0,\dots,\Delta_{n-1})^{\mathsf T}.
\] So it suffices to check the formulas on the basis $\varphi_m=\mathbf 1_{\PP^m}$, $m=0,\dots,n$, of linear subspaces. Since $\PP^m$ is smooth, $\CC(\varphi_m)=(-1)^m[\Con\PP^m]$, and $\Con\PP^m\cong\PP^m\times\PP^{n-m-1}$ has $\delta_i=1$ for $i=m$ and $0$ otherwise (for $m=n$ we get the zero section). So the coefficient vector of $\varphi_m$ is $(-1)^me_m$. On the other side, $\chi_k(\varphi_m)=\chi(\PP^{m+k-n})=\max(0,m+k-n+1)$, a ramp in $k$ with its kink at $k=n-m-1$. The matrix $[\chi_k(\varphi_m)]$ is triangular with respect to $m+k$ with ones on the antidiagonal $m+k=n$, so $A$ is invertible over $\Z$ and the coefficient vector is a universal linear function of the slice vector. The second difference of the ramp, centred at $k=n-i-1$, is $1$ exactly when $i=m$; multiplied by $(-1)^i$ it gives $(-1)^me_m$, as required. Finally $\chi_0(\varphi_m)=1$ exactly when $m=n$, which gives $n_0=(-1)^n\chi_0$.
\end{proof}

\begin{remark}[Smooth check]\label{rem:smoothcheck}
Let $X$ be smooth of degree $d$. Then $\CC(\mathbf 1_X)=(-1)^{n-1}[\Con X]$ and $\delta_i(\Con X)=d(d-1)^{n-1-i}$, so Theorem~\ref{thm:transform} asserts the classical identity
\[
\chi_k-2\chi_{k-1}+\chi_{k-2}=(-1)^{k-1}d(d-1)^{k-1}
\]
for the Euler characteristics $\chi_k$ of smooth hypersurfaces of degree $d$ in $\PP^k$. At $k=2$: $(3d-d^2)-2d+0=-d(d-1)$. Appendix~\ref{app:examples} runs the singular checks.
\end{remark}

\subsection{The ceiling}\label{sec:ceilingthm}

\begin{theorem}[Invariant ceiling]\label{thm:ceiling}
Let $f$ be nonzero homogeneous of degree $d\ge2$ in $N\ge3$ variables and $X=V(f)_{\mathrm{red}}\subset\PP^{n}$, and write $(-1)^{n-1}\CC(\mathbf{1}_X)=\sum_Sm_S[\Con\bar S]$ as in \eqref{eq:positivity}. For every stratum $S$ and every $0\le i\le n-1$,
\[
m_S\,\delta_i(\Con\bar S)\;\le\;\bigl|\Delta_i(\mathbf{1}_X)\bigr|\;\le\;4\,E(d,n-i)\;\le\;8(d-1)^{N-1}+8N .
\]
Consequently $\bigl(m_S\,\delta_i(\Con\bar S)\bigr)^{1/(N-1)}\le(d-1)\bigl(1+o(1)\bigr)$ as $N\to\infty$, uniformly in $d\ge2$.
\end{theorem}

\begin{proof}
All $m_S$ and all $\delta_i(\Con\bar S)$ are nonnegative, so $\sum_Sm_S\delta_i(\Con\bar S)=|\Delta_i(\mathbf{1}_X)|$ bounds every summand. By Theorem~\ref{thm:transform}, $\Delta_i$ involves only $\chi_{n-i},\chi_{n-i-1},\chi_{n-i-2}$, each bounded by $E(d,n-i)$ (Proposition~\ref{prop:envelope}; $E$ is nondecreasing in $k$), with coefficients of total size $4$. For the per-root statement: if $(d-1)^{N-1}\ge N$ the bound is at most $16(d-1)^{N-1}$, and otherwise at most $16N$.
\end{proof}

The same statement holds with $\mathbf{1}_X$ replaced by $\Eu_X$, whose characteristic cycle is a single conormal term, and, with masses $2^{O(N)}$ that disappear after taking roots, for the Chern--Mather and Chern--Schwartz--MacPherson classes, which are $\Z$-combinations of the same data \cite{Mac74,Alu18}.

\subsection{The nearby-cycle variant}\label{sec:phivariant}

\begin{lemma}[Reduced inputs with isolated singularities]\label{lem:nu-isolated}
Let $f$ be reduced of degree $d$ with $\dim\Sing X\le0$. Then $|\chi_k(\nu_f)|\le E(d,k)$ for every $0\le k\le n$, and the conclusion of Theorem~\ref{thm:ceiling} holds for $\CC(\nu_f)$, and for $\CC(\varphi_f)$ with twice the constant.
\end{lemma}

\begin{proof}
For $k<n$, a generic $L_k$ misses the finite set $\Sing X$, so $\nu_f\equiv1$ on $X\cap L_k$ and $\chi_k(\nu_f)=\chi(X\cap L_k)$. For $k=n$, take a generic form $g$ of degree $d$ and the pencil $X_t=V(f+tg)$, with total space $Y\subset\PP^n\times\C$ and base locus $B=X\cap V(g)$, which misses $\Sing X$. Properness of $Y\to\C$ gives $\chi(X_t)=\int_X\nu_{\mathrm{fam}}\,d\chi$ for small $t\ne0$, where $\nu_{\mathrm{fam}}(x)$ is the Euler characteristic of the Milnor fibre of $t|_Y$ at $x$. Off $B$, locally $t=-f/g$ with $g$ a unit, so $\nu_{\mathrm{fam}}=\nu_f$ there, including at the singular points of $X$. At $b\in B$, $X$ is smooth, the total space is smooth and $t|_Y$ is submersive, so $\nu_{\mathrm{fam}}(b)=1=\nu_f(b)$. Hence $\chi_n(\nu_f)=\chi(X_t)$ for a hypersurface $X_t$ of degree $d$, and Proposition~\ref{prop:envelope} applies. Positivity of $(-1)^{n-1}\CC(\nu_f)$ (Section~\ref{sec:nearby}) replaces \eqref{eq:positivity} in the proof of Theorem~\ref{thm:ceiling}.
\end{proof}

\begin{remark}[The one conditional constant]\label{rem:massey}
For general $f$ --- non-isolated singularities or non-reduced structure --- the same shape of envelope for $\chi_k(\nu_f)$ is expected: the multiplicity weights are absorbed exactly (for instance $\chi_1(\nu_f)=\sum_ie_i\deg g_i=d$ for $f=\prod g_i^{e_i}$; Appendix~\ref{app:nonreduced}), and the contributions along positive-dimensional singular strata are controlled by L\^e numbers, for which Massey's polar-type bounds \cite{Mas95} give envelopes of order $(d-1)^{\bullet}$. This yields $|\chi_k(\nu_f)|\le C_\nu^{\,k}E(d,k)$ for an absolute constant $C_\nu\ge1$, and we write $C_\nu$ for the best such constant. We have checked this on the non-reduced examples of Appendix~\ref{app:nonreduced} but do not give a complete proof. It is the single conditional ingredient of the corank-one theory, it enters only when Milnor-class ($\nu$-type) invariants of such inputs are used, and it affects only the constant of Theorem~\ref{thm:nogo}, never an exponent. For reduced inputs with isolated singularities $C_\nu=1$ by Lemma~\ref{lem:nu-isolated}.
\end{remark}

\section{The ceiling theorem}\label{sec:nogo}

We now combine the two halves. Section~\ref{sec:ceiling} bounds every characteristic-cycle invariant of the input by about $d-1$ per root, and Section~\ref{sec:count} bounds every valid solution count from below by about $4e\,m/N$ per root. This section makes ``every invariant'' and ``every count'' precise and proves Theorem~\ref{thm:intro-nogo}.

\subsection{The class of invariants}

\begin{definition}[Profile; invariants of characteristic-cycle type]\label{def:invariant}
Let $G$ be a nonzero homogeneous polynomial in $N'$ variables and $X=V(G)_{\mathrm{red}}\subset\PP^{N'-1}$. The \emph{profile} $\mathcal P(G)$ consists of the following numbers, for each $\varphi\in\{\mathbf 1_X,\ \Eu_X,\ \nu_G\}$ and each slot $0\le i\le N'-2$: the aggregated multidegrees $\Delta_i(\varphi)$, and the component data $n_W(\varphi)\,\delta_i(\Con W)$. Each number carries a slot. An invariant $I$ is \emph{of characteristic-cycle type} if $I(G)=|\Lambda(\mathcal P(G))|$ for a $\Z$-linear combination $\Lambda$ of finitely many profile numbers, selected by a rule that depends only on $N'$. It is in the \emph{$\mathbf 1_X$-scope} if $\Lambda$ involves only the $\mathbf 1_X$- and $\Eu_X$-blocks, and in the \emph{$\nu$-scope} otherwise.
\end{definition}

\begin{lemma}[Polar degrees and B\'ezout]\label{lem:polarbez}
For every reduced hypersurface $X\subset\PP^n$ of degree $d\ge2$ and every $0\le i\le n-1$, $\delta_i(\Con X)\le d(d-1)^{\,n-1-i}$.
\end{lemma}

\begin{proof}
By Bertini, $\delta_i$ is the number of points of $\Con X$ on $i$ generic divisors of class $h$ and $n-1-i$ generic divisors of class $\hd$, all lying over $X_{\mathrm{sm}}$. There the hyperplane is the tangent hyperplane, and the condition that it contains a point $p$ reads $\sum_kp^{(k)}\partial_kf(x)=0$, an equation of degree $d-1$. So the points counted are isolated solutions, in a generic $\PP^{n-i}$, of one equation of degree $d$ and $n-1-i$ equations of degree $d-1$, and refined B\'ezout \cite[$\S$12.3]{Ful98} bounds their number by $d(d-1)^{n-1-i}$ (compare \cite{Pie78,Tei73}).
\end{proof}

\begin{corollary}[Slotwise ceiling]\label{cor:slotwise}
With $f,X$ as in Theorem~\ref{thm:ceiling}, every stratum $S$ and every slot $i$ satisfy $m_S\,\delta_i(\Con\bar S)\le|\Delta_i(\mathbf 1_X)|\le4\,E(d,n-i)$, and $\delta_i(\Con X)\le d(d-1)^{n-1-i}$. The same holds for $\CC(\nu_f)$ with $4C_\nu^{\,n-i}E(d,n-i)$ on the right (Lemma~\ref{lem:nu-isolated} and Remark~\ref{rem:massey}).
\end{corollary}

\begin{proof}
$\Delta_i$ involves only $\chi_k$ with $k\le n-i$, and $E(d,k)$ is nondecreasing in $k$; positivity converts the total into a componentwise bound.
\end{proof}

\begin{lemma}[The classical invariants are in the class]\label{lem:invclass}
The following are of characteristic-cycle type: every single profile number, in particular every polar degree and every stratum datum $m_S\delta_i$; the degrees of the Chern--Mather class of $X$, which by Piene's formula are signed binomial combinations of polar degrees \cite{Pie78}; the degrees of the Chern--Schwartz--MacPherson class \cite{Mac74}, which are binomial combinations of generic-slice Euler characteristics \cite{Alu18} and hence, by Theorem~\ref{thm:transform}, of the profile; and the degrees of the Milnor class of $V(G)$ \cite{Alu18,Mas95}, in the $\nu$-scope. In each case the coefficients have total size at most $e^{O(N')}$.
\end{lemma}

\subsection{Readings}

Throughout this subsection an affine-linear representation of $f$ of size $m$ is identified, as in Remark~\ref{rem:homog}, with a homogeneous linear representation of $G=x_0^{m-d}f$ in $N'=N+1$ variables. Slots and codegrees refer to $\PP^{N'-1}=\PP^N$.

\begin{definition}[Valid reading; derivable bound]\label{def:reading}
Fix $N'\ge3$, a slot $0\le i\le N'-2$ and a size $m$. The \emph{slot-$i$ corank-one system} of a linear matrix $M$ in $N'$ variables is the incidence $Z_1$ intersected with $i$ generic linear conditions on $x$ (class $h$) and $N'-2-i$ generic contraction conditions (class $a+b$); for corank $c\ge2$ one uses $Z_c$ and the corresponding tautological globally generated classes on the Grassmann factors. A function $B_i(m)$ is a \emph{valid bound} for this system if for every $M$ of size $m$ its number of isolated solutions is at most $B_i(m)$. An invariant $I$ is \emph{read} through the system if
\[
I(x_0^{m-d}f)\;\le\;\#\{\text{isolated solutions of the system of }M\}
\]
for every size-$m$ representation $M$ of $f$; we write $I_m(f):=I(x_0^{m-d}f)$. A \emph{profile-matched reading} of $I=|\sum_i\lambda_i\,(\text{slot-}i\text{ numbers})|$ uses the bound $B(m)=\sum_i|\lambda_i|\,B_i(m)$ with each $B_i$ valid for the slot-$i$ system. The reading certifies $\dc(f)>m$ whenever $I_m(f)>B(m)$, and its \emph{derivable bound} is
\[
\mu_{I,B}(f)\;:=\;\min\bigl\{m_0:\ B(m)\ge I_m(f)\ \text{for all}\ m\ge m_0\bigr\};
\]
every lower bound on $\dc(f)$ that the reading certifies is smaller than $\mu_{I,B}(f)$.
\end{definition}

In words, a valid reading is an argument of the kind (a)--(c) of Theorem~\ref{thm:intro-nogo}: an injection of invariant data into the solutions of a kernel system, and a count bound valid for \emph{every} matrix of the given size. The proof of Theorem~\ref{thm:polar-exact} is a reading of $\delta_1(\Con V(G))=\delta_{\mathrm{top}}(V(f))$ at slot $1$. An argument that reads conormal data through a different injection, or in the reverse direction, is outside Definition~\ref{def:reading} even if it is conormal in spirit; Remark~\ref{rem:scope} lists what falls outside. By Proposition~\ref{prop:witness} and Theorem~\ref{thm:T}, every valid bound for the slot-$i$ corank-one system satisfies
\begin{equation}\label{eq:floor}
B_i(m)\;\ge\;T_i(N',m)\;=\;T(N'-i,m)\;=\;T(j+1,m),\qquad j:=N'-1-i ,
\end{equation}
where $j$ is the codegree of the slot.

\subsection{The homogenization multiplicity}

\begin{lemma}[The profile of the homogenized form]\label{lem:homog-profile}
Let $f$ be nonzero homogeneous of degree $d\ge2$ in $N\ge3$ variables, $m\ge d$, and $G=x_0^{m-d}f$ on $\PP^N$. For a slot of $\PP^N$ of codegree $j\in\{1,\dots,N\}$, every profile number of $G$ in that slot is at most $m$ if $j=1$, and for $j\ge2$ it is at most
\[
4E(d+1,j)\ \ \text{in the $\mathbf 1_X$- and $\Eu_X$-blocks,}\qquad
8C_\nu^{\,j}E(d+1,j)+4(m-d)E(d,j-1)\ \ \text{in the $\nu$-block,}
\]
where $C_\nu=1$ if $f$ is reduced and $V(f)\subset\PP^{N-1}$ has at most isolated singularities, and $C_\nu$ is the constant of Remark~\ref{rem:massey} otherwise.
\end{lemma}

\begin{proof}
\emph{The $\mathbf 1_X$- and $\Eu_X$-blocks.} $X=V(G)_{\mathrm{red}}$ is a reduced hypersurface of degree at most $d+1$ in $\PP^N$, so Corollary~\ref{cor:slotwise} and Lemma~\ref{lem:polarbez}, applied to $X$ and to its components, give $4E(d+1,j)$. In codegree $1$ the total is $|\Delta(\mathbf 1_X)|=\chi_1(\mathbf 1_X)=\deg X\le d+1\le m$, and similarly for $\Eu_X$.

\emph{The $\nu$-block.} Put $a=m-d$, $H=\{x_0=0\}$, $X_H=V(f)\cap H$, and let $C=V(f)\subset\PP^N$, the cone over $V(f)\subset\PP^{N-1}$ with vertex $o=[1:0:\dots:0]$. We first compute $\nu_G$. At $x\notin H$, $x_0$ is a unit and $\nu_G(x)=\nu^C(x)$, the Milnor function of $f$ on $C$. At $x\in H$ the local model is $G=z^ag(w)$ with $z=x_0$ and $g$ a local equation of $f$; the Milnor fibre $\{z^ag(w)=\varepsilon\}$ is an $a$-sheeted cover of $\{w:g(w)\ne0\}$, so its Euler characteristic is $a\bigl(\chi(B_w)-\chi(B_w\cap V(g))\bigr)$, which is $0$ if $x\in X_H$ and $a$ otherwise. Hence
\begin{equation}\label{eq:taut}
\nu_G\;=\;\nu^{\mathrm{prop}}+a\bigl(\mathbf 1_H-\mathbf 1_{X_H}\bigr),\qquad \nu^{\mathrm{prop}}:=\nu^C\cdot\mathbf 1_{C\setminus H}.
\end{equation}
Away from $o$ and $H$ the cone is locally the product of the germ of $V(f)$ at $\pi(x)$ with a line, where $\pi$ is the projection from $o$ to $H\cong\PP^{N-1}$, so $\nu^{\mathrm{prop}}=\nu_f\circ\pi$ there; at the vertex, $\nu^{\mathrm{prop}}(o)=\chi(\{f=1\})$, since the Milnor fibre of a homogeneous polynomial at the origin is its global affine fibre.

Now the slices. For $k\le N-1$, a generic $L_k$ misses $o$, $\pi$ maps it isomorphically onto a generic $k$-plane of $H$, and $L_k\cap H$ is a generic $(k-1)$-plane of $H$; so $\chi_k(\nu^{\mathrm{prop}})=\chi_k(\nu_f)-\chi_{k-1}(\nu_f)$, computed in $\PP^{N-1}$. For $k=N$, $\chi_N(\nu^{\mathrm{prop}})=\chi(\{f=1\})+\chi(\C^*)\cdot(\cdots)=\chi(\{f=1\})$, which is at most $2d(d-1)^{N-1}\le E(d+1,N)$ in absolute value by the proof of Proposition~\ref{prop:envelope}. Using $|\chi_k(\nu_f)|\le C_\nu^{\,k}E(d,k)$ (Lemma~\ref{lem:nu-isolated}, Remark~\ref{rem:massey}) we get $|\chi_k(\nu^{\mathrm{prop}})|\le 2C_\nu^{\,k}E(d+1,k)$ for all $k$. Next, $\chi_k(\mathbf 1_{X_H})$ is the Euler characteristic of a slice of $V(f)$ by a generic $(k-1)$-plane of $H$, so $|\chi_k(\mathbf 1_{X_H})|\le E(d,k-1)$; and $\CC(\mathbf 1_H)=\pm[\Con H]$ has a single nonzero multidegree, in codegree $1$. By Theorem~\ref{thm:transform} and linearity, for $j\ge2$
\[
|\Delta(\nu_G)|\;\le\;4\cdot 2C_\nu^{\,j}E(d+1,j)+a\cdot 4E(d,j-1)
\]
in the slot of codegree $j$. Since $(-1)^{N-1}\CC(\nu_G)$ is effective, each component is at most this total. In codegree $1$ the total is $\chi_1(\nu_G)$, the sum of $\nu_G$ over the points of a generic line on $V(G)$; at a generic point of a component $\nu_G$ is its multiplicity, so the total is $\deg G=m$.
\end{proof}

\subsection{The theorem}

\begin{theorem}[Ceiling]\label{thm:nogo}
Let $f$ be nonzero homogeneous of degree $d\ge2$ in $N\ge3$ variables, and let $I$ be an invariant of characteristic-cycle type. Every profile-matched corank-one valid reading of $I$ (Definition~\ref{def:reading}) satisfies
\[
\mu_{I,B}(f)\;\le\;\bigl\lceil K(d+1)(N+1)\bigr\rceil\;\le\;3K\,dN,\qquad K:=\max(48,\,7C_\nu),
\]
where $C_\nu=1$ unless $I$ is in the $\nu$-scope and $f$ is non-reduced or $V(f)$ has non-isolated singularities, in which case $C_\nu$ is the constant of Remark~\ref{rem:massey}.
\end{theorem}

\begin{proof}
Let $m\ge m^\ast:=\lceil K(d+1)(N+1)\rceil$. We show that in every slot the valid bound dominates every profile number of $G=x_0^{m-d}f$ in that slot. For a profile-matched reading of $I=|\sum_i\lambda_i(\cdot)|$ this gives
\[
B(m)\;=\;\sum_i|\lambda_i|\,B_i(m)\;\ge\;\sum_i|\lambda_i|\,\bigl|\text{slot-}i\text{ number}\bigr|\;\ge\;I_m(f),
\]
the weights cancelling on both sides, and hence $\mu_{I,B}(f)\le m^\ast$.

\emph{Codegree $1$.} By \eqref{eq:floor}, $B_i(m)\ge T(2,m)=m$, and by Lemma~\ref{lem:homog-profile} every profile number in this slot is at most $m$.

\emph{Codegree $j\ge2$.} Since $m\ge 2(N+1)\ge2(j+1)$, Proposition~\ref{prop:floor} and \eqref{eq:floor} give $B_i(m)\ge Q^j$ with $Q:=2m/(j+1)\ge 2Kd$, the last inequality because $m\ge Kd(N+1)\ge Kd(j+1)$. On the other side, $E(d+1,j)\le 6jd^j$ and $E(d,j-1)\le 4jd^{j-1}$, so by Lemma~\ref{lem:homog-profile} every profile number in this slot is at most $48C_\nu^{\,j}jd^j+16\,m\,jd^{j-1}$. Now
\[
\tfrac12Q^j\;\ge\;\tfrac12(2Kd)^j\;\ge\;48\,C_\nu^{\,j}\,j\,d^j
\quad\text{and}\quad
\tfrac12Q^j\;\ge\;\frac{m}{j+1}\,(2Kd)^{j-1}\;\ge\;16\,m\,j\,d^{j-1},
\]
because $(2K/C_\nu)^j\ge 96j$ and $(2K)^{j-1}\ge16j(j+1)$ for all $j\ge2$: both hold at $j=2$ by the choice of $K$ ($4K^2\ge192C_\nu^2$ and $2K\ge96$), and the left sides grow faster in $j$ than the right sides. Adding the two inequalities gives the claim. Finally $(d+1)(N+1)\le\frac94dN$, so $m^\ast\le3KdN$.
\end{proof}

\begin{remark}[Higher corank]\label{rem:highercorank}
For corank $c\ge2$ we do not have a complete argument. Proposition~\ref{prop:monotone} shows that, under Hypothesis~\ref{hyp:S}, a slot-$0$ corank-$c$ count is at least the corank-one count at size $m-c+1$, which is what the proof above would need in that slot.
\end{remark}

\begin{remark}[Restricting first does not help]\label{rem:restrict}
If $g$ is a linear restriction of $f$ in $N'\le N$ variables, then $\deg g\le d$, so Theorem~\ref{thm:nogo} applied to $g$ shows that a reading of an invariant of $g$ certifies at most $3K\,dN'$. In particular the smooth-slice bounds of Section~\ref{sec:slices} are optimal for the slices they use, up to a constant factor.
\end{remark}

\begin{corollary}[The power sum]\label{cor:diagonal}
Let $d=N\to\infty$ and $f=\sum_{i=1}^Nx_i^N$. Every bound on $\dc(f)$ derivable as in Theorem~\ref{thm:nogo} is $O(N^2)$, and for invariants in the $\mathbf 1_X$-scope
\[
\mu_{I,B}(f)\;\le\;(1+o(1))\,\frac{N^2}{4e},
\]
while \eqref{eq:positive} achieves $(1-o(1))N^2/(4e)$. So the reach of the method on the power sum is $(1\pm o(1))N^2/(4e)$ in the $\mathbf 1_X$-scope, and $\Theta(N^2)$ in general.
\end{corollary}

\begin{proof}
The first claim is Theorem~\ref{thm:nogo}. For the second, fix $\varepsilon\in(0,1)$, let $m\ge m^\ast:=\lceil(1+\varepsilon)(N+1)d/(4e)\rceil$, and work with $G=x_0^{m-d}f$ in $\PP^N$, whose profile numbers in the $\mathbf 1_X$-scope are at most $4E(d+1,j)\le24\,j\,d^j$ in codegree $j\ge2$ and at most $d+1\le m$ in codegree $1$ (Lemma~\ref{lem:homog-profile}); these bounds do not depend on $m$. For codegree $j\ge2$ with $j+1\le(N+1)/(48e)$, Proposition~\ref{prop:floor} gives $B_i(m)\ge(2m/(j+1))^j\ge(96e\,m/(N+1))^j\ge(24(1+\varepsilon)d)^j\ge24\,j\,d^j$. For $j+1>(N+1)/(48e)$, both $j$ and $m/(j+1)\ge m/(N+1)$ tend to infinity, so Theorem~\ref{thm:asymp} gives $B_i(m)^{1/j}\ge(4e-o(1))\,m/(j+1)\ge(1+\varepsilon)(1-o(1))\,d$, which exceeds $(24jd^j)^{1/j}=(1+o(1))d$ for large $N$. As in the proof of Theorem~\ref{thm:nogo}, $B(m)\ge I_m(f)$ for all $m\ge m^\ast$, and we let $\varepsilon\to0$.
\end{proof}

The order $O(N^2)$ is not informative by itself, since $\dc(\sum_ix_i^N)=O(N^2)$ (Remark~\ref{rem:max}); the content of the corollary is the constant. The true value of $\dc(\sum_ix_i^N)$ may exceed $N^2/(4e)$, and the corollary shows that the method cannot certify any such excess. Per root, no polynomial of degree $N$ in $N$ variables gives the method more to work with than the power sum.

\begin{corollary}[The permanent]\label{cor:perm}
Let $n$ be the size of the permanent, so that $\per_n$ is homogeneous of degree $n$ in $n^2$ variables. Every lower bound on $\dc(\per_n)$ derivable as in Theorem~\ref{thm:nogo} is $O(n^3)$. More generally, for any family whose degree and number of variables are polynomial in $n$, every such bound is polynomial in $n$; in particular, within the scope of Theorem~\ref{thm:nogo}, the method cannot prove $\dc(\per_n)=n^{\omega(\log n)}$.
\end{corollary}

\begin{proof}
Theorem~\ref{thm:nogo} with $d=n$ and $N=n^2$ gives $\mu_{I,B}(\per_n)\le3Kn^3$.
\end{proof}

The cubic lower bound of \cite{OAI26} for $\dcb(\per_n)$ is built on kernel equations and a multihomogeneous count of the kind studied here, so for the permanent the order of Corollary~\ref{cor:perm} is attained, and what remains open for arguments of this type is the constant.

\begin{remark}[Symmetric representations]\label{rem:symm}
For symmetric representations the bound of Theorem~\ref{thm:smooth-slice} has the larger constant $1/(2e)$. This does not conflict with Corollary~\ref{cor:diagonal}: a symmetric representation is read through the symmetric kernel incidence, in which kernel and cokernel coincide, which is not the system of Definition~\ref{def:reading}, and a lower bound for $\sdcb$ is not a lower bound for $\dcb$.
\end{remark}

\begin{remark}[A cap for $\delta_{\mathrm{top}}$ beyond intersection theory]\label{rem:max}
For the single invariant $\delta_{\mathrm{top}}$, the cap holds for \emph{any} monotone use: an argument of the shape ``$\delta_{\mathrm{top}}(f)>\sup\{\delta_{\mathrm{top}}(g):\dc(g)\le m\}\Rightarrow\dc(f)>m$'' is bounded by the determinantal complexity of any near-maximizer of $\delta_{\mathrm{top}}$ in degree $d$. The Fermat polynomial $\sum_{i\le N}x_i^d$ has $\delta_{\mathrm{top}}=d(d-1)^{N-2}$, the maximum allowed by Lemma~\ref{lem:polarbez}, and $\dc(\sum_ix_i^d)=O(dN)$, because it is computed by a formula with $O(dN)$ nodes and a formula of size $s$ is the determinant of a matrix of size $s+2$ \cite{Val79}. The content of Theorem~\ref{thm:nogo} is that passing to the full profile, which is strictly finer (Proposition~\ref{prop:witness-pair}), does not raise the cap.
\end{remark}

\begin{remark}[Scope]\label{rem:scope}
Definition~\ref{def:reading} covers every argument that extracts an invariant as in Definition~\ref{def:invariant} and compares it with a solution count of a corank-one kernel system bounded uniformly over all matrices. It does not cover scheme-theoretic invariants of the conormal or dual variety, such as the Hilbert function or the regularity of $X^\vee$, for which no bound in terms of $m$ is known in either direction, nor mechanisms that do not factor through solution counts of kernel systems at all.
\end{remark}

\section{Comparison with other methods}\label{sec:barriers}

\subsection{The full profile is strictly finer than \texorpdfstring{$\delta_{\mathrm{top}}$}{the class}}

\begin{proposition}[Witness pair]\label{prop:witness-pair}
There are irreducible plane quartics $C_1$ with exactly three nodes and $C_2$ with exactly two ordinary cusps. They satisfy $\delta_{\mathrm{top}}(C_1)=12-2\cdot3=6=12-3\cdot2=\delta_{\mathrm{top}}(C_2)$, but $-\CC(\mathbf 1_{C_1})=[\Con C_1]+\sum_{p}[\Con p]$ over the three nodes and $-\CC(\mathbf 1_{C_2})=[\Con C_2]+\sum_{p}[\Con p]$ over the two cusps, with slot-$0$ totals $9$ and $8$; and the vanishing-cycle masses are $\{1,1,1\}$ and $\{2,2\}$.
\end{proposition}

\begin{proof}
$\delta_{\mathrm{top}}$ is the Pl\"ucker class $d(d-1)-2\#\mathrm{nodes}-3\#\mathrm{cusps}$ \cite{Pie78,GKZ94}; the genus formula $g=3-\#\mathrm{nodes}-\#\mathrm{cusps}$ gives $0$ and $1$, and both curves exist. Nodes and cusps of plane curves have multiplicity $2$, hence local Euler obstruction $2$, so $\mathbf 1_C=\Eu_C-\sum_p\mathbf 1_p$ and each singular point contributes $[\Con p]$, which has $\delta_0=1$. The Milnor numbers are $\mu(\text{node})=1$ and $\mu(\text{cusp})=2$, and $\varphi_f(p)=-\mu(p)$ for plane curves (the Milnor fibres $\{xy=\varepsilon\}\simeq\C^\ast$ and $\{x^2-y^3=\varepsilon\}$, a wedge of two circles up to homotopy).
\end{proof}

So a method that reads the full profile has strictly more information than one that reads $\delta_{\mathrm{top}}$, and Theorem~\ref{thm:nogo} is correspondingly stronger than a cap on $\delta_{\mathrm{top}}$-arguments. Finer invariants are not larger numbers, though, and Corollary~\ref{cor:slotwise} bounds each of them slot by slot.

\subsection{Dimension arguments}
Landsberg, Manivel and Ressayre \cite{LMR13} obtain their bound for the permanent through the dimension of the dual variety. This dimension can be read off the profile --- $\delta_i(\Con X)=0$ for $i<\codim X^\vee-1$ and $\delta_i\ne0$ at the threshold --- so it is a coarsening of the profile, but it is not a linear function of it, and the argument of \cite{LMR13} is not a reading in the sense of Definition~\ref{def:reading}; Theorem~\ref{thm:nogo} does not apply to it. The same holds for the mechanism of Alper, Bogart and Velasco \cite{ABV17}, which uses the codimension of $\Sing V(f)$, i.e.\ the codimension of the largest stratum other than $X$ carrying mass in \eqref{eq:positivity}; their bounds are linear in $N$. The codimension principle (Theorem~\ref{thm:codim}) converts the same quantity into a bound of the form $\Omega(c\,d)$ by way of a smooth slice.

\subsection{Rank methods}
Shifted partial derivatives and flattening ranks are not functions of the characteristic cycle: all smooth hypersurfaces of degree $d$ in $\PP^{N-1}$ have the same profile, while their spaces of partial derivatives can have very different dimensions (compare the Fermat polynomial with a generic form). So these methods lie outside Definition~\ref{def:invariant}, and they are capped by the independent barriers of \cite{EGOW18,ELSW18}. The two kinds of cap come from unrelated arguments.

\subsection{Hessian rank}
The method of \cite{MR04} bounds $\dc(P)$ by half the rank of the Hessian at a suitable point, so it never exceeds half the number of variables. For $H_q$, with $q^2$ variables, Theorem~\ref{thm:hom-lower} gives $(q-1)(q^2-1)/(4e)$, which exceeds $q^2/2$ exactly when $(q-1)(q^2-1)>2eq^2$, that is for all $q\ge7$. From $q=7$ on, the bound for $H_q$ is beyond the reach of Hessian-rank arguments, even for $\dc$.

\subsection{Circulant slices of the permanent}

Another way to look for a smooth slice of the permanent is to restrict to circulant matrices. Let $\Circ_n(y)$ be the $n\times n$ matrix with $(i,j)$ entry $y_{j-i}$, indices modulo $n$, and $P_n(y_0,\dots,y_{n-1})=\per(\Circ_n(y))$.

\begin{proposition}\label{prop:circ-composite}
If $n$ is composite, then $V(P_n)\subseteq\PP^{n-1}$ is singular.
\end{proposition}

\begin{proof}
Let $2\le t<n$ divide $n$, and restrict to the coordinate subspace $S$ where $y_j=0$ unless $t\mid j$. After permuting rows and columns by residue classes modulo $t$, the circulant matrix becomes block diagonal with $t$ identical blocks, so $P_n|_S=Q^t$ for a nonzero form $Q$ of degree $n/t\ge2$, which has zeros in $\PP(S)$. At a zero $p$ of $Q$, the derivatives with respect to the retained variables vanish because $t\ge2$. For $y_j$ with $t\nmid j$, give each $y_k$ the weight $k\bmod t$; every monomial $\prod_iy_{\sigma(i)-i}$ of $P_n$ has weight $\sum_i(\sigma(i)-i)\equiv0$, so every monomial of $\partial P_n/\partial y_j$ has weight $-j\not\equiv0$ and contains a variable that vanishes on $S$. So the whole gradient vanishes at $p$.
\end{proof}

\begin{conjecture}\label{conj:prime-circ}
If $p$ is prime, then $V(P_p)\subseteq\PP^{p-1}$ is smooth.
\end{conjecture}

The conjecture holds for $p=2,3,5$: $P_2=y_0^2+y_1^2$; $P_3=y_0^3+y_1^3+y_2^3+3y_0y_1y_2$ is a smooth member of the Hesse pencil; and for $p=5$ a Gr\"obner basis computation shows that the partial derivatives have no common nonzero zero (Appendix~\ref{app:comp}). If it holds for all primes, then for every $n$ we can choose a prime $p$ with $n/2<p\le n$, restrict $\per_n$ to $\per_p$ by setting the complementary diagonal block to the identity, and restrict further to $P_p$; Theorem~\ref{thm:smooth-slice} then gives $\dcb(\per_n)\ge(p-1)^2/(4e)=\Omega(n^2)$, an independent route to a quadratic bound.

\section{Open problems}\label{sec:open}

\begin{enumerate}
\item Prove or disprove Conjecture~\ref{conj:topdim}. The first case is $q=3$: does the ideal generated by the nine partial derivatives of $H_3$ define a set of dimension $3$ in $\PP^8$? A positive answer for all $q$ would show that the ceiling of Theorem~\ref{thm:nogo} is attained, in order, by a $\VNP$-complete family.
\item Find smooth linear slices of $H_q$ in more than $q$ variables. By Theorem~\ref{thm:coord-max} they must use non-coordinate linear forms.
\item Settle Conjecture~\ref{conj:prime-circ}.
\item Find a natural family of homomorphism polynomials that is $\VNP$-complete under $p$-projections and has a large Fermat slice without added loops.
\item Determine the best constant for the permanent: what is the largest $\kappa$ such that arguments covered by Theorem~\ref{thm:nogo} can prove $\dcb(\per_n)\ge(\kappa-o(1))n^3$?
\item Remove the conditional constant of Remark~\ref{rem:massey}, and prove Hypothesis~\ref{hyp:S} or find a different treatment of higher-corank kernel systems.
\item Find lower-bound arguments outside Definition~\ref{def:reading}, for instance using scheme-theoretic invariants of the dual variety. By Theorem~\ref{thm:asymp}, any improvement beyond $O(dN)$ must change \emph{what} is read, not how well the present reading is executed.
\end{enumerate}

\appendix
\section{Worked examples}\label{app:examples}

Every numerical claim in this appendix was checked by hand and by exact integer or rational arithmetic; the conic example was also confirmed by a Gr\"obner-basis computation.

\subsection{The conic: two polar points and one extra B\'ezout point}\label{app:conic}

Take $N=3$, $m=2$ and $M(x)=\begin{pmatrix}x_1&x_2\\x_2&x_3\end{pmatrix}$, so $\det M=x_1x_3-x_2^2$ is a smooth conic, with $T(3,2)=2$ and $\Bbez(3,2)=3$. In the chart $x_3=v_2=u_2=1$ the kernel equations are
\[
x_1v_1+x_2=0,\quad x_2v_1+1=0,\quad u_1x_1+x_2=0,\quad u_1x_2+1=0,
\]
with solutions $x=(v_1^{-2},-v_1^{-1},1)$, $u_1=v_1$: the incidence is a rational curve. The contraction condition $c_1u_1v_1+c_2(u_1+v_1)+c_3=0$ at $c=(3,5,7)$ becomes $3v_1^2+10v_1+7=0$, so $v_1\in\{-1,-\tfrac73\}$, giving exactly the two polar points
\[
\bigl((1,1,1);(-1,1);(-1,1)\bigr),\qquad\bigl((9,21,49);(-\tfrac73,1);(-\tfrac73,1)\bigr),
\]
the second on the conic since $9\cdot49=21^2$. If the second entry of $u^{\mathsf T}M=0$ is discarded, as in the B\'ezout count, there is one more solution, $x=(0,0,1)$, $v=(1,0)$, $u=(-\tfrac53,1)$: it satisfies $M(x)v=0$, the kept equation $u_1x_1+x_2=0$ and the contraction $3u_1+5=0$, but not $u^{\mathsf T}M(x)=0$, whose second entry is $1$. This is the excess that the kernel-bundle class $c_{m-1}(K)$ of Theorem~\ref{thm:class} removes.

\subsection{A nodal cubic}

For an irreducible nodal plane cubic the slice vector is $(\chi_0,\chi_1,\chi_2)=(0,3,1)$, so Theorem~\ref{thm:transform} gives $(\Delta_0,\Delta_1)=(1-6,\,-3)=(-5,-3)$. Independently, $-\CC(\mathbf 1_X)=[\Con X]+[\Con p]$ with $\delta(\Con X)=(4,3)$ (Pl\"ucker: $3\cdot2-2=4$) and $\delta(\Con p)=(1,0)$, with totals $(5,3)$.

\subsection{Three general planes in \texorpdfstring{$\PP^3$}{P3}}

The slice vector is $(0,3,3,4)$ ($\chi_2$: a triangle of lines, $3\cdot2-3$; $\chi_3$: $9-6+1$), and the transform gives $\Delta=(4-6+3,\,-(3-6),\,3)=(1,3,3)$. On the other hand $\mathbf 1_X=\sum\mathbf 1_{P_i}-\sum\mathbf 1_{L_{ij}}+\mathbf 1_q$ gives $\CC(\mathbf 1_X)=\sum[\Con P_i]+\sum[\Con L_{ij}]+[\Con q]$, with multidegree vectors $e_2,e_1,e_0$, and $3e_2+3e_1+e_0=(1,3,3)$. All masses are $+1$, as positivity requires.

\subsection{The Whitney umbrella in \texorpdfstring{$\PP^3$}{P3}}

Let $X=V(F)$ with $F=x_0x_1^2-x_2^2x_3$, a cubic surface singular along the line $L=\{x_1=x_2=0\}$, with pinch points $p_1=[1{:}0{:}0{:}0]$ and $p_2=[0{:}0{:}0{:}1]$.

\emph{Polar degrees.} $\delta_2=\deg X=3$. The Gauss map is $\gamma(x)=[x_1^2:2x_0x_1:-2x_2x_3:-x_2^2]=:[y_0:y_1:y_2:y_3]$. On $X$, $x_0y_0=-x_3y_3$, and substituting gives $y_0y_3(y_0y_1^2+y_2^2y_3)=0$ on the image, so $X^\vee\subseteq V(y_0y_1^2+y_2^2y_3)$. This cubic is irreducible and the Gauss map is generically finite, so $X^\vee=V(y_0y_1^2+y_2^2y_3)$ (the umbrella is self-dual up to coordinates) and $\delta_0=3$. For $\delta_1$, let $\Gamma_q$ be the polar curve of a generic point $q$; as $1$-cycles $X\cdot V(D_qF)=[\Gamma_q]+2[L]$, because $F\in(x_1,x_2)^2$ while $D_qF$ has order exactly one along $L$ and the transverse-slice ideal $(u^2-tv^2,\,2q_1u-2q_2tv)$ has colength $2$ at a generic point of $L$. So $\delta_1=\deg\Gamma_q=6-2=4$, and $\delta(\Con X)=(3,4,3)$.

\emph{Euler characteristics.} On the chart $x_0=1$ the affine umbrella $\{x_1^2=x_2^2x_3\}$ is the image of $(s,t)\mapsto(st,s,t^2)$, which is one-to-one off the handle $\{x_1=x_2=0\}$, two-to-one over the punctured handle and one-to-one over the origin; additivity gives $\chi=1$ for the affine part. At infinity, $X\cap\{x_0=0\}$ is two lines meeting in a point, $\chi=3$. So $\chi(X)=4$. A generic plane section is an irreducible cubic with one node, so $\chi_2=1$, and $\chi_1=3$, $\chi_0=0$: the slice vector is $(0,3,1,4)$.

\emph{Assembly.} The transform gives $\Delta=(4-2+3,\,-(1-6),\,3)=(5,5,3)$. Subtracting $\delta(\Con X)=(3,4,3)$ leaves $(2,1,0)=2[\Con\mathrm{pt}]+[\Con L]$: mass $1$ on the singular line and total mass $2$ on the pinch points, which are exchanged by the involution $(x_0,x_1)\leftrightarrow(x_3,x_2)$ and so carry mass $1$ each. L\^e--Teissier confirms this: $\Eu_X(\text{pinch})=2-1=1$, so $\mathbf 1_X=\Eu_X-\Eu_L+\mathbf 1_{p_1}+\mathbf 1_{p_2}$.

\subsection{Non-reduced examples}\label{app:nonreduced}

\textbf{(a)} $f=x_0^2x_1$ in $\PP^2$, $d=3$, $X_{\mathrm{red}}=L_0\cup L_1$, $p=L_0\cap L_1$. The Milnor fibre function is $2$ on $L_0\setminus p$, $1$ on $L_1\setminus p$, and $0$ at $p$ (the fibre $\{z^2w=\varepsilon\}$ is a double cover of a punctured disc), so $\nu_f=2\cdot\mathbf 1_{L_0}+\mathbf 1_{L_1}-3\cdot\mathbf 1_p$, with slice vector $(0,3,3)$ --- note $\chi_1=d$, the multiplicity-weighted degree. The transform gives $\Delta(\nu_f)=(-3,-3)$, matching $\CC(\nu_f)=-2[\Con L_0]-[\Con L_1]-3[\Con p]$; the sign $(-1)^{n-1}=-1$ makes all masses positive.

\textbf{(b)} $f=g^2$ with $g$ a smooth conic, $d=4$. Then $\nu_f\equiv2$ on the conic, the slice vector is $(0,4,4)$, and $\Delta(\nu_f)=(-4,-4)=2\,\Delta(\mathbf 1_{V(g)})$: the multiplicity scales the profile and is bounded by the degree of the given equation.

\subsection{The determinantal quintic threefold and small values}

For $(N,m)=(5,5)$ the generic determinantal quintic $X_M\subset\PP^4$ has $50$ ordinary nodes (Proposition~\ref{prop:anchors}), and Teissier's formula gives $\delta_{\mathrm{top}}=5\cdot4^3-2\cdot50=220=T(5,5)$. Small values of $T(N,m)$:
\[
\begin{array}{c|cccccc}
N\backslash m&2&3&4&5&6&7\\\hline
3&2&6&12&20&30&42\\
4&2&12&36&80&150&252\\
5&0&12&68&220&540&1120\\
6&0&6&84&430&1440&3780
\end{array}
\]
The rows $N=3,4$ are the smooth dual degrees $m(m-1)$ and $m(m-1)^2$. The entry $T(5,2)=0$ is a check: the determinant of a generic $2\times2$ linear matrix in five variables is a rank-$4$ quadric in $\PP^4$, a cone, whose dual is not a hypersurface.

\section{Computations for Section~\ref{sec:hom}}\label{app:comp}

None of the main lower bounds depends on these calculations.

\paragraph{Symbolic identities.}
The identity for $H_2$ in Theorem~\ref{thm:irreducible} and the formulas in Proposition~\ref{prop:h3-ranktwo} were checked in SymPy by expanding $H_2$ and $H_3$ from the definition, as sums over all pairs of maps $\alpha,\beta:[q]\to[q]$ ($16$ pairs for $q=2$ and $729$ for $q=3$), and then substituting and differentiating.

\paragraph{Circulant permanents.}
For $p=3$ and $p=5$ we computed a Gr\"obner basis, in graded reverse lexicographic order, of the ideal generated by the partial derivatives of $P_p$. In both cases the leading monomials include a pure power of every variable, so the ideal is zero-dimensional; since it is homogeneous, its only zero is the origin, and $V(P_p)$ is smooth.

\paragraph{Points of $\Sing V(H_3)$ over finite fields.}
Over $\FF_5$ and $\FF_7$ we enumerated all $9$-tuples $X\in\FF_q^{3\times3}$ and evaluated the nine partial derivatives of $H_3$ using Lemma~\ref{lem:moment}. Since $9$ is invertible in these fields, Euler's identity $9H_3=\sum x_{ij}\,\partial H_3/\partial x_{ij}$ shows that the common zeros of the derivatives form the affine cone over $\Sing V(H_3)$. The counts, with projective points sorted by rank, are:
\[
\begin{array}{c|c|c|ccc}
q & \text{affine zeros} & \text{projective points} & \text{rank }1 & \text{rank }2 & \text{rank }3\\
\hline
5 & 3577 & 894 & 336 & 540 & 18\\
7 & 26083 & 4347 & 945 & 3402 & 0
\end{array}
\]
The rank-one counts agree with Theorem~\ref{thm:rankone}: over $\FF_5$ the Fermat cubic curve has $6$ points and $\PP^2$ has $31$, giving $2\cdot6\cdot31-6^2=336$, and over $\FF_7$ the count is $2\cdot9\cdot57-9^2=945$. There are many singular points of rank two over these fields, but the counts do not determine the dimension of $\Sing V(H_3)$ over $\C$.

\paragraph{A random linear section.}
We also tested the linear subspace $\PP^4\subseteq\PP^8$ given, in coordinates $t_0,\dots,t_4$, by
\[
X=\begin{pmatrix}
2t_2 & 3t_0 & 3t_3\\
4t_4 & 2t_1 & t_1\\
4t_2 & 4t_3 & 2t_4
\end{pmatrix}.
\]
Exhaustive enumeration over $\FF_5$ and over $\FF_{25}=\FF_5[i]/(i^2-2)$ found no nonzero common zero of the nine partial derivatives of $H_3$ on this subspace.

\section*{Acknowledgments and use of AI}

The author used generative AI systems, including ChatGPT, extensively in preparing this paper: as generators and independent verifiers of arguments, for drafting and revising proofs, for computer-algebra scripts, and for consistency checks. The author directed the work and checked the definitions, statements, proofs and computations, and all claims are meant to be checkable from the text.

\end{document}